\documentclass[12pt]{article}
\usepackage{latexsym, amsbsy, amsthm,amssymb,multirow, epsfig, amsmath}
\usepackage{hyperref}
\usepackage{setspace}
\usepackage{natbib,graphicx,setspace,lscape,longtable}
\usepackage{appendix}
\usepackage{bm}
\usepackage{enumitem}
\usepackage{makecell}
\usepackage{threeparttable}
\usepackage{graphicx}
\usepackage{color}
\usepackage{xcolor}
\usepackage{subfigure}
\usepackage{array}
\RequirePackage[mathlines, displaymath]{lineno}
\usepackage{latexsym}
\usepackage{epsfig}
\usepackage{bm}
\usepackage{makecell}
\usepackage{threeparttable}
\usepackage{booktabs}
\usepackage{float}
\usepackage{tikz}
\usetikzlibrary{shapes,snakes}
\bibpunct{(}{)}{;}{a}{,}{,}
\usepackage{geometry}
\def\beqr{\begin{eqnarray}}
\def\eeqr{\end{eqnarray}}
\def\beqrs{\begin{eqnarray*}}
\def\eeqrs{\end{eqnarray*}}

\makeatletter

\newcommand{\Rmnum}[1]{\expandafter\@slowromancap\romannumeral #1@}
\makeatother

\def\1{{\bf 1}}

\newtheorem{prop}{\sc Proposition}

\begin{document}
	\title{An absolute measure of heterogeneity for lnOR in meta-analysis}
	
	{\small
		\author{Ke Yang$^{1}$, 
			Lei Shi$^{2}$, Wangli Xu$^{3}$, Liping Zhu$^{4}$, and Tiejun Tong$^{5,}$\thanks{Corresponding author. E-mail: tongt@hkbu.edu.hk}\\ \\
			{\small $^{1}$Department of Statistics and Data Science, Beijing
				University of Technology,}\\
			{\small Beijing, China}\\
			{\small $^{2}$School of Statistics and Mathematics, Yunnan University of Finance and Economics,}\\
			{\small Yunnan, China}\\
			{\small $^{3}$Center for Applied Statistics and School of Statistics, Renmin University of China,}\\ {\small Beijing, China}\\
			{\small $^{4}$Center for Applied Statistics and Institute of Statistics and Big Data,}\\
			{\small Renmin University of China, Beijing, China}\\
			{\small $^{5}$Department of Mathematics, Hong Kong Baptist University, Hong Kong, China}
		}
	}
	\date{}
	\maketitle
	
	\begin{abstract}
		Quantifying the heterogeneity is an important issue in meta-analysis. The conventional heterogeneity statistic, $I^2$, heavily depends on study sample sizes and tends to approach one as the sample sizes increase. To overcome this limitation, a sample size independent measure, ${\rm ICC}_{\rm MA}$, has recently been developed to quantify the heterogeneity, yet this method is restricted to meta-analysis with continuous outcomes. In practice, however, binary outcomes are more widely encountered, with the log odds ratio (lnOR) frequently adopted as an effect size to compare the effects of two treatments in terms of odds.  
		To fill the gap, we propose a new heterogeneity measure for meta-analysis with binary outcomes by introducing a latent continuous variable model. Under this framework, the proposed measure admits a closed-form expression within the random-effects meta-analysis model with lnOR as the effect size. We define this sample size independent measure as ${\rm ICC}_{\rm MA}^{\rm OR}$, and moreover propose a new estimator, $I^2_A$, to evaluate this measure. Simulation studies demonstrate that our new estimator is nearly unbiased and aligns closely with its counterpart applied for standardized mean differences when the continuous variable is observable. A real data application further demonstrates that the proposed heterogeneity statistic for binary outcomes and its counterpart for continuous outcomes provide interpretable and coherent estimates across different outcome types.
		
		\vskip 12pt
		\noindent
		$Key \ words$: Absolute measure of heterogeneity; latent variable model; log odds ratio; meta-analysis; standardized mean difference.
		
	\end{abstract}

\newpage
\baselineskip 22pt
\setlength{\parskip}{0.2\baselineskip}

\section{Introduction}

Meta-analysis has become the standard statistical approach for synthesizing evidence from multiple independent studies and plays a fundamental role in evidence-based medicine \citep{egger1997meta,higgins2019cochrane}. Besides estimating an overall treatment effect, an equally important objective is to quantify the extent to which treatment effects vary across studies. Such between-study heterogeneity directly influences the interpretation, generalizability, and clinical applicability of pooled estimates. When the heterogeneity is negligible, the pooled effect adequately represents the treatment effect across studies, whereas substantial heterogeneity indicates genuine differences among study populations. Consequently, reliable quantification of the between-study heterogeneity has remained one of the central methodological problems in meta-analysis.

Under the random-effects model (REM), the between-study heterogeneity is naturally characterized by the between-study variance, $\tau^2$. Since the magnitude of $\tau^2$ depends on the scale of the effect size, direct comparisons across different meta-analyses are often difficult. To overcome this limitation, \cite{higgins2002quantifying} and \cite{higgins2003measuring} proposed the widely used heterogeneity statistic, $I^2$, which quantifies the proportion of total variability attributable to the between-study heterogeneity rather than within-study sampling error \citep{higgins2026}. Owing to its intuitive interpretation and ease of computation, $I^2$ has become the most widely reported heterogeneity statistic in evidence synthesis. Despite its widespread use, $I^2$ has an important limitation: it depends on the within-study sampling variances and therefore on study sample sizes. Consequently, $I^2$ may approach one as study sample sizes increase, even when the underlying between-study heterogeneity remains unchanged \citep{rucker2008undue,riley2016external,borenstein2017basics}. Therefore, $I^2$ reflects the relative variability of the observed effect estimates rather than the absolute variability among the underlying study populations.

To advance the literature,  \cite{yang2025alternative} reconsidered the heterogeneity from the perspective of the underlying individual-level populations rather than the observed study-level effect estimates. Motivated by the intraclass correlation coefficient in analysis of variance (ANOVA), they introduced the absolute heterogeneity measure
\begin{equation}
	\label{eq:iccma}
	{\rm ICC}_{\rm MA}
	=
	\frac{\tau^2}
	{\tau^2+\sigma_{\rm pop}^2},
\end{equation}
where $\sigma_{\rm pop}^2$ denotes a representative population variance of the individual-level outcome. ${\rm ICC}_{\rm MA}$ is invariant to study sample sizes and directly quantifies the proportion of total population variation attributable to differences between study populations. They further proposed the heterogeneity statistic $I_{\rm A}^2$, which consistently estimates ${\rm ICC}_{\rm MA}$ using only study-level summary statistics. Consequently, ${\rm ICC}_{\rm MA}$ provides an absolute and clinically interpretable measure of heterogeneity. Nevertheless, their new method is limited to meta-analysis with continuous outcomes, leaving an important methodological gap for the studies reporting binary outcomes.

Extending ${\rm ICC}_{\rm MA}$ to binary outcomes is considerably more challenging than simply replacing one effect size by another. The difficulty arises because, unlike continuous outcomes, binary responses do not possess a natural individual-level population variance, which underlies the definition of ${\rm ICC}_{\rm MA}$. This challenge can be overcome by viewing binary outcomes as dichotomized realizations of latent continuous variables. Such latent-variable formulations have long been used to establish the relationship between standardized mean differences (SMDs) and log odds ratios (lnORs) \citep{cox1989analysis,chinn2000simple,murad2019continuous,higgins2019cochrane}. Consequently, they provide a natural foundation for extending the population-level interpretation of heterogeneity from continuous to binary outcomes.

Inspired by this, this paper develops an absolute heterogeneity measure for meta-analysis with binary outcomes. We establish a latent-variable framework linking binary outcomes and their underlying continuous counterparts at the individual level. Under a logistic latent-variable model, we derive an explicit relationship between study-specific lnORs and SMDs and show that the corresponding random-effects models are intrinsically connected. This connection naturally extends the population-level interpretation underlying ${\rm ICC}_{\rm MA}$ from meta-analysis with continuous outcomes to that with binary outcomes, leading to a sample size invariant heterogeneity measure, denoted by ${\rm ICC}_{\rm MA}^{\rm OR}$, together with its practical estimator $I_{\rm A}^2$.

The proposed framework offers three main advantages. First, it provides a unified population-level interpretation of heterogeneity for meta-analysis with both continuous and binary outcomes, enabling coherent comparisons across different outcome types. Second, unlike the conventional $I^2$ statistic, the proposed measure is invariant to study sample sizes and therefore quantifies the absolute heterogeneity among study populations. Third, the proposed methodology requires only published study-level summary statistics and can be readily implemented using existing estimators of the between-study variance, facilitating its application in routine evidence synthesis.

The remainder of this paper is organized as follows. Section 2 presents a motivating example of meta-analysis  involving binary and continuous outcomes, highlighting the connection between the two types of outcomes. Section 3 develops a latent continuous variable model
for binary outcomes, leading to the absolute heterogeneity measure ${\rm ICC}_{\rm MA}^{\rm OR}$ for lnOR and its
estimator $I_{\rm A}^2$. Section 4 evaluates the performance of the proposed statistic via simulation studies and by revisiting the real data example in Section 2. Lastly, Section 5 concludes with a summary and discusses potential avenues for future research.

\section{A motivating example}

Following \cite{yang2025alternative}, subsequent studies have increasingly recognized the reliability of ${\rm ICC}_{\rm MA}$ as a measure of heterogeneity and have applied its estimator $I_{\rm A}^2$ as a heterogeneity statistic in practice. For example, \cite{hong2025paul} stated, ``\textit{We used the $I_{\rm A}^2$ statistic as a measure for quantifying the heterogeneity in our meta-analysis, as it is an unbiased estimator that also has the property of sample size invariance  distinguishing it from the traditional Higgins $I^2$ statistic.}"

To fill the important gap that calls for the development of a counterpart for binary outcomes, we now present an example from the network meta-analysis of \cite{zhou2020comparative}, where some studies report both continuous and binary outcomes for the same treatment. This setting suggests a potential connection between binary and continuous outcomes, indicating that the heterogeneity measure developed for continuous outcomes may inform its binary counterparts.
The dataset includes 71 randomized controlled trials on treatments for depressive disorders in children and adolescents. Among these, seven trials compare fluoxetine with placebo in patients with major depressive disorder, of which five report both binary and continuous outcomes based on the same sample and are retained in our analysis.
Data are extracted from \cite{zhou2020comparative}, \cite{cipriani2016comparative}, and the original trial reports, and are summarized in Table \ref{tab:data_summary}. Each study consisted of two arms: a treatment group receiving fluoxetine and a control group receiving placebo. For each arm, we record the sample size, the continuous outcome of mean score reduction from baseline with its standard deviation, and the binary outcome representing the number of responders.
\begin{table}
	\centering
	\caption{Summary of the Five Fluoxetine vs Placebo Studies}
	\label{tab:data_summary}
	\begin{tabular}{lcccccc}
		\hline
		& \multicolumn{3}{c}{Treatment Group (Fluoxetine)} & \multicolumn{3}{c}{Control Group (Placebo)} \\
		\cline{2-4} \cline{5-7}\\[-6pt]
		Study&  $N$ & Reduction (SD) & Responder & $N$ & Reduction (SD) & Responder \\
		\hline
		Almeida(2005)&7&9.14(4.02)&4& 9&8.14(4.54)&6\\
		Atkinson(2014)&113& 23.7(11.27)&71& 103&24.3(11.27)&64\\
		Emslie(1997)&48&20.1(13.19)&27& 48&10.5(14.84)&16\\
		Emslie(2014)&112&22.6(13.44)&68& 117&21.6(13.74)&70\\
		Weihs(2018)& 110&24.8(12.27)& 86& 112&23.1(12.49)&70\\
		\hline
	\end{tabular}
\end{table}

To clarify the measurement of efficacy outcomes, Table \ref{measurementtools} summarizes the assessment instruments. Continuous outcomes were measured by either the Children's Depression Rating Scale–Revised (CDRS-R; range 17–113) or the Depression Self-Rating Scale for Children (DSRS; range 0–36), both widely applied to assess depressive symptom severity.
Binary outcomes were defined in two ways. The first dichotomizes the continuous measure, defining response as at least a 50\% rreduction in the CDRS-R score from baseline, which can be interpreted as a threshold on symptom reduction given comparable baseline scores across studies. The second employed the Clinical Global Impressions-Improvement (CGI-I) scale, a 7-point clinician-rated measure of global improvement, with responders defined as CGI-I scores of 1 or 2. Although CGI-I is not directly derived from the continuous scales, it reflects overall clinical improvement and is thus conceptually linked to the continuous outcomes.
\begin{table}
	\centering
	\caption{Measurement Tools and Binary Outcome Definitions in the Five Studies}
	\label{measurementtools}
	\begin{tabular}{lll}
		\hline
		Study & Continuous Outcome& Binary Outcome\\
		\hline
		Almeida (2005) & DSRS & Not known \\
		Atkinson (2014) & CDRS-R & 50\% reduction from baseline on CDRS-R \\
		Emslie(1997)& CDRS-R & CGI-I rating of 1 or 2 \\
		Emslie (2014) & CDRS-R & 50\% reduction from baseline on CDRS-R \\
		Weihs (2018) & CDRS-R & CGI-I rating of 1 or 2 \\
		\hline
	\end{tabular}
\end{table}

Taken together, these observations suggest that binary outcomes can be viewed as thresholded forms of underlying latent continuous variables. Several studies have leveraged this framework to establish transformations between lnOR and SMD, enabling their joint synthesis (\citealp{cox1989analysis,suissa1991binary,whitehead1999combining,chinn2000simple,murad2019continuous,jing2023bayesian}). Motivated by this, we develop the counterpart of ${\rm ICC}_{\rm MA}$ in (\ref{eq:iccma}) to meta-analysis with lnOR as the effect size using a latent continuous variable model for individual-level data, as formalized in Section 3.

\section{Heterogeneity measure for lnOR}

The key difficulty in extending the absolute heterogeneity measure from continuous to binary outcomes is that binary responses do not possess a natural population variance at the individual level. Consequently, the denominator of ICC$_{\rm MA}$ cannot be directly defined. To overcome this difficulty, we introduce a latent continuous variable formulation, under which binary outcomes arise through thresholding of an underlying continuous response. This framework naturally links REMs for SMD and lnOR, providing the basis for extending the population-level interpretation of heterogeneity from continuous to binary outcomes.

\subsection{Individual-level latent variable formulation}
To develop a unified heterogeneity measure, we first introduce individual-level formulations for continuous and binary outcomes. We emphasize that, at this stage, the two outcome types are modeled separately. Their connection through a latent variable formulation will be established in the next subsection.

For the $i$th study $(i=1,\ldots,k)$, let $\{\tilde y_{ij}^{T}\}_{j=1}^{n_i^T}$ and $\{\tilde y_{ij'}^{C}\}_{j'=1}^{n_i^C}$ denote the continuous responses for the $j$th individual from the treatment and the $j'$th individual from the control groups, respectively. We assume that the individual-level responses satisfy
\beqr\label{latent_model}
\tilde y^T_{ij}=\sigma_i(\tilde\mu_i^T+\xi_{ij}^T) \qquad \text{and} \qquad\tilde y^C_{ij'}=\sigma_i(\tilde\mu_i^C+\xi_{ij'}^C),
\eeqr
where $\sigma_i>0$ is a study-specific scale parameter accounting for differences in measurement scales across studies, $\tilde\mu_i^T$ and $\tilde\mu_i^C$ are the standardized group means, and the independent random errors $\xi_{ij}^T$ and $\xi_{ij'}^C$ have mean zero and unit variance.

The corresponding study-specific treatment effect is measured by SMD,
\beqr\label{tmu_i}
\tilde\mu_i = \tilde\mu_i^T -\tilde\mu_i^C.
\eeqr
At the study level, the observed SMDs $\{\tilde y_i\}_{i=1}^k$ are synthesized using REM as
\begin{equation}\label{smd_rem}
	\tilde y_i = \tilde\mu_i + \tilde\epsilon_i,\quad
	\tilde\mu_i \stackrel{\text{i.i.d.}}{\sim} N(\tilde\mu, \tilde\tau^2), \quad
	\tilde\epsilon_i \stackrel{\text{ind}}{\sim} N(0, \sigma_{\tilde y_i}^2),
\end{equation}
where $\tilde\mu$ denotes the overall treatment effect, $\tilde\tau^2$ is the between-study variance, and $\tilde\epsilon_i$ represents the within-study sampling error with variance $\sigma_{\tilde y_i}^2$.

We next consider binary outcomes measured on the same individuals. For the $i$th study, let
$\{y_{ij}^{T}\}_{j=1}^{n_i^T}$ and
$\{y_{ij'}^{C}\}_{j'=1}^{n_i^C}$
denote the binary responses for the treatment and control groups, respectively, where
$y_{ij}^{T}$ and $y_{ij'}^{C}$ correspond to the same individuals whose latent continuous responses are
$\tilde y_{ij}^{T}$ and $\tilde y_{ij'}^{C}$ introduced above. Each binary response takes value one if the event of interest occurs and zero otherwise. Denote  by $p_i^T$ and $p_i^C$  the event probabilities in the treatment and control groups, respectively, for the $i$th study. Given $p_i^T$, the individual outcomes $y_{ij}^T$ for $j = 1, \dots, n_i^T$ are assumed to be independent and follow Bernoulli distributions. Similarly, given $p_i^C$, the individual outcomes $y_{ij'}^C$ for $j' = 1, \dots, n_i^C$ are assumed to be independent and follow Bernoulli distributions. 
That is,
\beqr\label{prob}
y_{ij}^T|p_i^T \stackrel{\text{i.i.d.}}{\sim}  \mathrm{Bernoulli}(p_i^T) \qquad \text{and} \qquad
y_{ij'}^C|p_i^C \stackrel{\text{i.i.d.}}{\sim}  \mathrm{Bernoulli}(p_i^C).
\eeqr
The study level treatment effect is summarized by lnOR,
\beqr  \label{mu_i}
\mu_i=\ln\frac{p_i^T/(1-p_i^T)}{p_i^C/(1-p_i^C)}.
\eeqr

At the study level, we adopt REM to synthesize the observed lnORs $\{y_i\}_{i=1}^k$,
\begin{equation}\label{lnorrem}
	y_i = \mu_i + \epsilon_i,\quad
	\mu_i \stackrel{\text{i.i.d.}}{\sim} N(\mu, \tau^2), \quad
	\epsilon_i \stackrel{\text{ind}}{\sim} N(0, \sigma_{y_i}^2),
\end{equation}
where $\mu$ denotes the overall treatment effect, $\tau^2$ is the between-study variance, and $\epsilon_i$ represents the within-study sampling error with variance $\sigma_{y_i}^2$.

The continuous and binary outcomes introduced above are defined on the same underlying individuals but are modeled separately, leading to REMs (\ref{smd_rem}) and (\ref{lnorrem}), respectively. In the next subsection, we specify a latent variable mechanism linking the continuous and binary responses at the individual level. Under this formulation, we show that the two random-effects models are intrinsically connected, thereby providing the theoretical basis for extending the proposed heterogeneity measure from SMD to lnOR.

\subsection{Linking REMs for continuous and binary outcomes}

Although REMs in (\ref{smd_rem}) and (\ref{lnorrem}) are formulated separately, they may be intrinsically connected if the binary outcomes arise from underlying latent continuous responses. We therefore introduce a latent variable mechanism linking the two outcome types and investigate the relationship between their corresponding random-effects models.

To achieve this goal, we relate binary outcomes $y^T_{ij}$, $y^C_{ij'}$ in (\ref{prob}) to underlying continuous variables $\tilde y^T_{ij}$, $\tilde y^C_{ij'}$ in (\ref{latent_model}) through a threshold mechanism. Without loss of generality, we define
\beqr\label{trunc}
y^T_{ij}=\begin{cases}
	1, & \text{if } \tilde y^T_{ij}>C_i,\\
	0, & \text{otherwise},
\end{cases}
\qquad\quad
y^C_{ij'}=\begin{cases}
	1, & \text{if } \tilde y^C_{ij'}>C_i,\\
	0, & \text{otherwise},
\end{cases}
\eeqr
where $C_i$ is a study-specific threshold. Under this formulation, continuous and binary outcomes are represented within a common individual-level framework, allowing their corresponding REMs to be compared directly.

The following proposition establishes the relationship between the two REMs in (\ref{smd_rem}) and (\ref{lnorrem}), with the proof given in Appendix A.
\begin{prop}\label{prop1}
	Suppose the binary outcomes satisfy model (\ref{prob}). We assume that there exists underlying latent continuous variables in (\ref{latent_model}), such that the observed binary outcomes are generated from the latent variables via a threshold mechanism (\ref{trunc}). Assume that the error term in (\ref{latent_model}) follows a logistic distribution with mean 0 and variance 1.
	The study-level effect sizes $\mu_i$ in (\ref{mu_i}) for binary outcomes and $\tilde{\mu}_i$ in (\ref{tmu_i}) for latent continuous outcomes satisfy
	\beqr\label{logit}
	\mu_i = \frac{\pi}{\sqrt{3}} \tilde{\mu}_i.
	\eeqr
	Consequently, REM for SMD in (\ref{smd_rem}) induces the corresponding REM for lnOR in (\ref{lnorrem}) with
	\beqr\label{relat}
	\mu= \frac{\pi}{\sqrt{3}} \tilde{\mu}\, \, \, \text{and} \, \, \,
	\tau^2 = \frac{\pi^2}{3}\tilde\tau^2.
	\eeqr
\end{prop}

The relationship (\ref{logit}) follows from the logistic distribution assumed for the individual-level errors in (\ref{latent_model}) and is invariant to both the study-specific scale parameter $\sigma_i$ and the threshold $C_i$. Consequently, the transformation from SMD to lnOR depends only on the underlying latent distribution rather than the measurement scale or the choice of threshold.

Equation (\ref{logit}) is consistent with the classical conversion between SMD and lnOR proposed by \cite{chinn2000simple} and subsequently adopted in evidence synthesis \citep{higgins2019cochrane,murad2019continuous}. Proposition~\ref{prop1}, however, goes beyond this well-known conversion by establishing the corresponding relationship between the two random-effects meta-analysis models. In particular, both the overall treatment effect and the between-study variance are transformed through the same linear scaling factor. Consequently, any population-level heterogeneity measure defined through the between-study variance on the latent continuous scale can be transferred directly to the lnOR scale, providing the theoretical foundation for the heterogeneity measure developed in the next subsection.

The proposition does not depend on any particular specification of the study-specific thresholds. Furthermore, when the latent errors follow a standard normal distribution instead of a logistic distribution, the same relationship remains an accurate approximation. The corresponding derivation and numerical evaluation are presented in Appendix B.

\subsection{Absolute measure of heterogeneity for lnOR}
Given the intrinsic relationship between REMs for SMD and lnOR established in Proposition \ref{prop1}, we develop the counterpart of ${\rm ICC}_{\rm MA}$ for SMD in (\ref{smd_rem}) to REM for lnOR in (\ref{lnorrem}).

For the latent continuous outcomes in (\ref{latent_model}), the standardized error terms $\xi_{ij}^T$ and $\xi_{ij'}^C$ are assumed to have unit variance, implying
$\sigma_{\rm pop}^2=1$. Following \cite{yang2025alternative}, the heterogeneity measure for SMD in model (\ref{smd_rem}) is defined as
\beqr\label{lv}
{\rm ICC}_{\rm MA}^{\rm LV}=\frac{\tilde\tau^2}{\tilde\tau^2+1}.
\eeqr
Using the relationship in (\ref{relat}), we have $\tilde\tau^2 = 3\tau^2/\pi^2$. Substituting this expression into (\ref{lv}) yields the corresponding heterogeneity measure for (\ref{lnorrem}):
\beqr\label{iccor}
{\rm ICC}_{\rm MA}^{\rm OR}=\frac{\tau^2}{\tau^2+\pi^2/3}.
\eeqr

In (\ref{iccor}), the population variance in the denominator is a fixed constant. This follows from the latent variable formulation in (\ref{latent_model}) together with the threshold model in (\ref{trunc}), which jointly represent binary outcomes as dichotomized observations of an underlying standardized continuous process. Consequently, the variance of lnORs is fixed and independent of the observed study data. Compared with ${\rm ICC}_{\rm MA}^{\rm LV}$ in (\ref{lv}), whose population variance is 1, the corresponding variance for lnOR in (\ref{iccor}) is larger. This increase arises from the $\pi/\sqrt{3}$ scaling factor introduced when transforming the SMD effect size to the lnOR scale in (\ref{logit}).

Similarly to ${\rm ICC}_{\rm MA}$ proposed by \cite{yang2025alternative} for continuous outcomes, our new ${\rm ICC}_{\rm MA}^{\rm OR}$ has the following desirable properties:
\begin{enumerate}
	\item[a)] \textit{Monotonicity.} ${\rm ICC}_{\rm MA}^{\rm OR}$ increases monotonically with the between-study variance $\tau^2$.
	\item[b)] \textit{Location invariance.} ${\rm ICC}_{\rm MA}^{\rm OR}$ is unaffected by the location of the effect sizes $\mu_i$.
	\item[c)] \textit{Study size invariance.} ${\rm ICC}_{\rm MA}^{\rm OR}$ is not affected by the total number of studies $k$.
	\item[d)] \textit{Sample size invariance.} ${\rm ICC}_{\rm MA}^{\rm OR}$ is not affected by sample size of individual studies.
\end{enumerate}
The fourth property ensures that ${\rm ICC}_{\rm MA}^{\rm OR}$ is invariant to study sample sizes. It can therefore serve as an absolute measure of heterogeneity.

To facilitate interpretation, \cite{yang2026five} adopted the six-level grading scheme proposed by \cite{landis1977measurement} to describe the level of ${\rm ICC}_{\rm MA}$ for meta-analysis with continuous outcomes. We note that ${\rm ICC}_{\rm MA}^{\rm OR}$ has the same variance proportion interpretation, and so the same scheme can be extended to binary outcomes.  Specifically, ${\rm ICC}_{\rm MA}^{\rm OR}=0$ indicates no heterogeneity. Values of 0-0.2, 0.2-0.4, 0.4-0.6, 0.6-0.8, and 0.8-1 correspond to low, moderate, substantial, severe and extreme heterogeneity, respectively. 

\subsection{Estimation of ${\rm ICC}_{\rm MA}^{\rm OR}$}

This section proposes an estimator  of ${\rm ICC}_{\rm MA}^{\rm OR}$ for practical use. For convenience, we denote it by the same notation $I^2_{\rm A}$ as in \cite{yang2025alternative}. With the observed effect sizes $y_i$ and the within-study variances $\sigma^2_{y_i}$ in (\ref{lnorrem}), $I^2_{\rm A}$ is obtained by substituting an estimator of the between-study variance $\tau^2$ into the population-level expression of ${\rm ICC}_{\rm MA}^{\rm OR}$ in (\ref{iccor}). Accordingly, regardless of the estimation method adopted for the between-study variance $\tau^2$, including the DerSimonian–Laird (DL) estimator (\citealp{dersimonian1986meta}), the restricted maximum likelihood (REML) estimator (\citealp{patterson1971recovery,viechtbauer2005bias}), and the Paule–Mandel (PM) estimator (\cite{paule1982consensus}), $I^2_{\rm A}$ can be written in the unified plug-in form
\beqr\label{ia}
I^2_{\rm A}=\frac{\hat\tau^2}{\hat\tau^2+\pi^2/3},
\eeqr
where $\hat\tau^2$ denotes an estimator of the between-study variance $\tau^2$ in (\ref{lnorrem}).

Among estimators of $\tau^2$, the DL method is particularly attractive due to its closed-form expression, which also leads to a simple representation of $I^2_{\rm A}$ in terms of Cochran’s $Q$ statistic. When the DL method is adopted, the between-study variance $\tau^2$ can be estimated as
\beqr\label{tauh}
\hat\tau^2_{\rm DL}=\max\left\{0, \frac{Q-\left(k-1\right)}{\sum_{i=1}^k w_i-\sum_{i=1}^k w_i^2/\sum_{i=1}^k w_i}\right\},
\eeqr
where
\beqrs
Q=\sum_{i=1}^k  w_i\left( y_i-\frac{\sum_{i=1}^k  w_i y_i}{\sum_{i=1}^k  w_i}\right)^2
\eeqrs
is Cochran’s $Q$ statistic and $w_i=1/\sigma^2_{y_i}$ are the inverse-variance weights.
Rewriting the denominator in (\ref{tauh}) as $(k-1)\tilde w$, where $\tilde w=(\sum_{i=1}^k w_i-\sum_{i=1}^k w_i^2/\sum_{i=1}^k w_i)/(k-1)$, we have
\beqrs
\hat\tau^2_{\rm DL}=\max\left\{0, \frac{Q-\left(k-1\right)}{\left(k-1\right)\tilde w}\right\}.
\eeqrs
Substituting this expression into (\ref{ia}) yields
\beqr\label{iadl}
I^2_{\rm A}=\max\left\{0,\frac{Q-\left(k-1\right)}{Q+\left(k-1\right)\left(\tilde w\pi^2/3-1\right)}\right\},
\eeqr

For comparison, the conventional $I^2$ statistic is
\beqr{\label{i2}}
I^2=\max\left\{0,\frac{Q-\left(k-1\right)}{Q}\right\}.
\eeqr
Comparing $I^2$ in (\ref{i2}) with $I^2_{\rm A}$ in (\ref{iadl}), both statistics share the same numerator $Q-(k-1)$ but differ in their denominators. In particular, $I^2$ uses $Q$, whereas $I^2_{\rm A}$ includes an additional term $(k-1)\left(\tilde w\pi^2/3 -1\right)$. Since $\tilde w$ is the adjusted mean inverse-variance weight, it increases as within-study variances decrease, i.e., as study sample sizes grow and estimates become more precise. As a result, the denominator of $I^2_{\rm A}$ is larger than that of $I^2$, yielding systematically smaller values of $I^2_{\rm A}$. This indicates that the conventional $I^2$ may tend to overestimate the heterogeneity between study populations.

\section{Numerical results}

This section evaluates the performance and interpretability of the proposed heterogeneity statistic $I^2_A$ through simulation studies and demonstrates its application by the motivating example introduced in Section 2.

\subsection{Performance of $I^2_{\rm A}$ under different settings}

To investigate the performance of $I^2_A$ for binary outcomes using lnOR as the effect size, we conduct simulations under varying event rates, numbers of studies, and degrees of the between-study heterogeneity.

The between-study heterogeneity is set to $\tau^2=0.3$ and $3$, corresponding to ${\rm ICC}_{\rm MA}^{\rm OR}\approx0.084$ and $0.477$. For each study, $\mu_i^C \sim N(\mu^C,\tau^2/2)$ and $\mu_i^T \sim N(\mu^T,\tau^2/2)$, inducing $\mu_i=\mu_i^T-\mu_i^C$ with variance $\tau^2$. Event counts are generated from binomial models with $n=100$, and lnORs are computed from the resulting $2\times2$ tables. Random-effects meta-analyses are performed, with $\tau^2$ estimated by the DerSimonian–Laird estimator (DL), the restricted maximum likelihood estimator (REML), and the Paule–Mandel estimator (PM), respectively, and $I^2_A$ is computed based on $\hat{\tau}^2$ by formula (\ref{ia}).

We consider two levels of between-study heterogeneity, characterized by $\tau^2=0.3$ and $\tau^2=3$, corresponding to ${\rm ICC}_{\rm MA}^{\rm OR}=0.3/(0.3+\pi^2/3)\approx0.084$ and ${\rm ICC}_{\rm MA}^{\rm OR}=3/(3+\pi^2/3)\approx0.477$, respectively.
For each study $i=1,\ldots,k$, the study-specific log odds for the control and treatment groups are generated as $\mu_i^C \sim N(\mu^C,\tau^2/2)$ and $\mu_i^T \sim N(\mu^T,\tau^2/2)$. The resulting lnOR $\mu_i=\mu_i^T-\mu_i^C$ has between-study variance $\tau^2$.

Given the study-specific log odds $\mu_i^C$ and $\mu_i^T$, the corresponding event probabilities are obtained via the inverse logit transformation, $p_i^C=\{1+\exp(-\mu_i^C)\}^{-1}$ and $p_i^T=\{1+\exp(-\mu_i^T)\}^{-1}$. Event counts in each arm are then generated from binomial distributions with sample size $n=100$. The observed lnORs are computed from the resulting $2\times2$ tables, and random-effects meta-analysis is performed using the DL, REML, and PM estimators for $\tau^2$. The heterogeneity measure $I^2_{\rm A}$ in (\ref{ia}) is calculated based on each estimate $\hat\tau^2$.

For each combination of $(k, \tau^2, \mu^T, \mu^C)$, we conduct $M=1000$ replications. Results are summarized using Chauvenet-type boxplots (Lin et al., 2025) of the estimated $I^2_A$ across different control event rates $p^C$. Figure \ref{fig1} reports results for the null-effect scenario with $p^T=p^C$, while Figure \ref{fig2} corresponds to the positive-effect scenario with $p^T=p^C+0.1$. The corresponding true overall lnORs for $p^C=0.1,0.2,0.3,0.4,0.5$ are $0.811, 0.539, 0.442, 0.405,$ and $0.405$, respectively.
\begin{figure}
	\begin{center}
		\begin{tabular}{cc}
			\psfig{figure=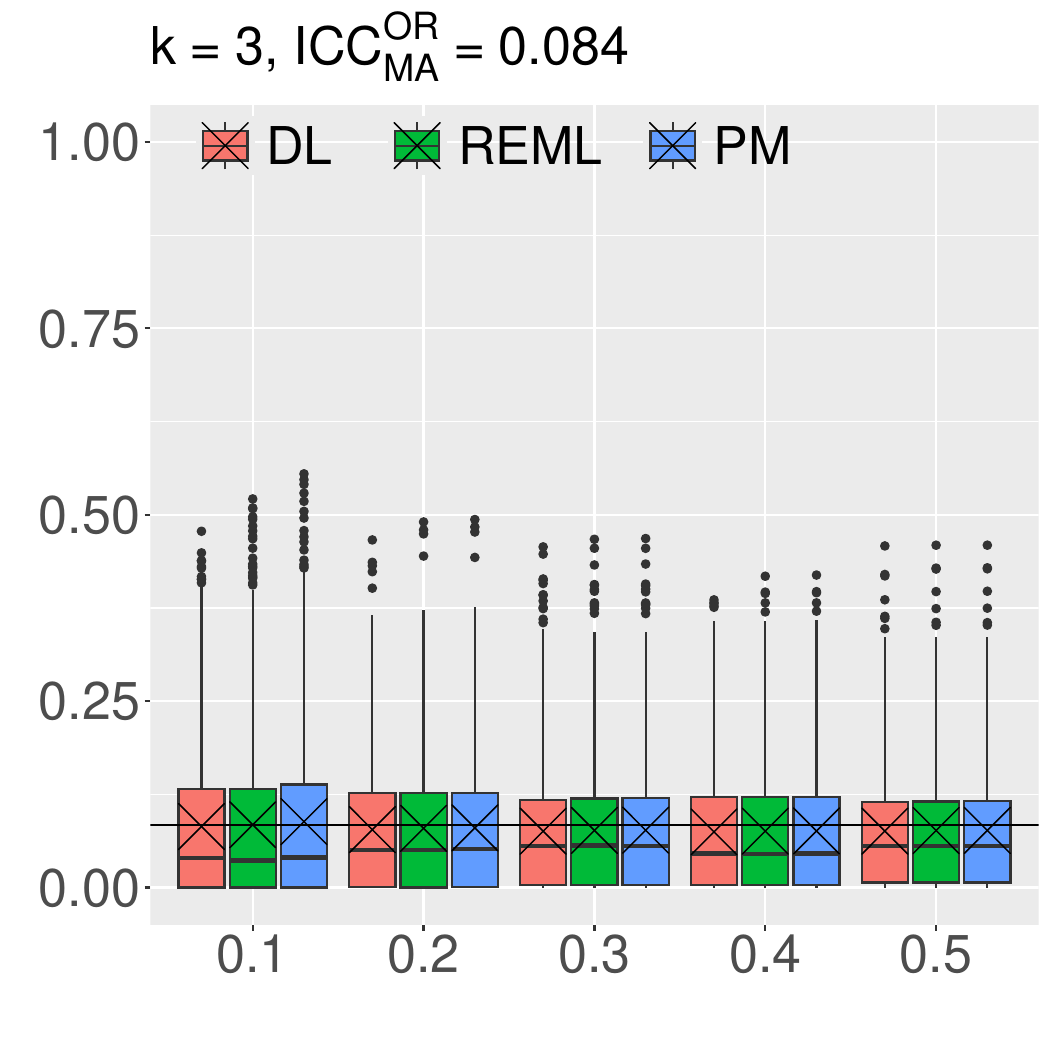,width=3.0in,angle=0}&
			\psfig{figure=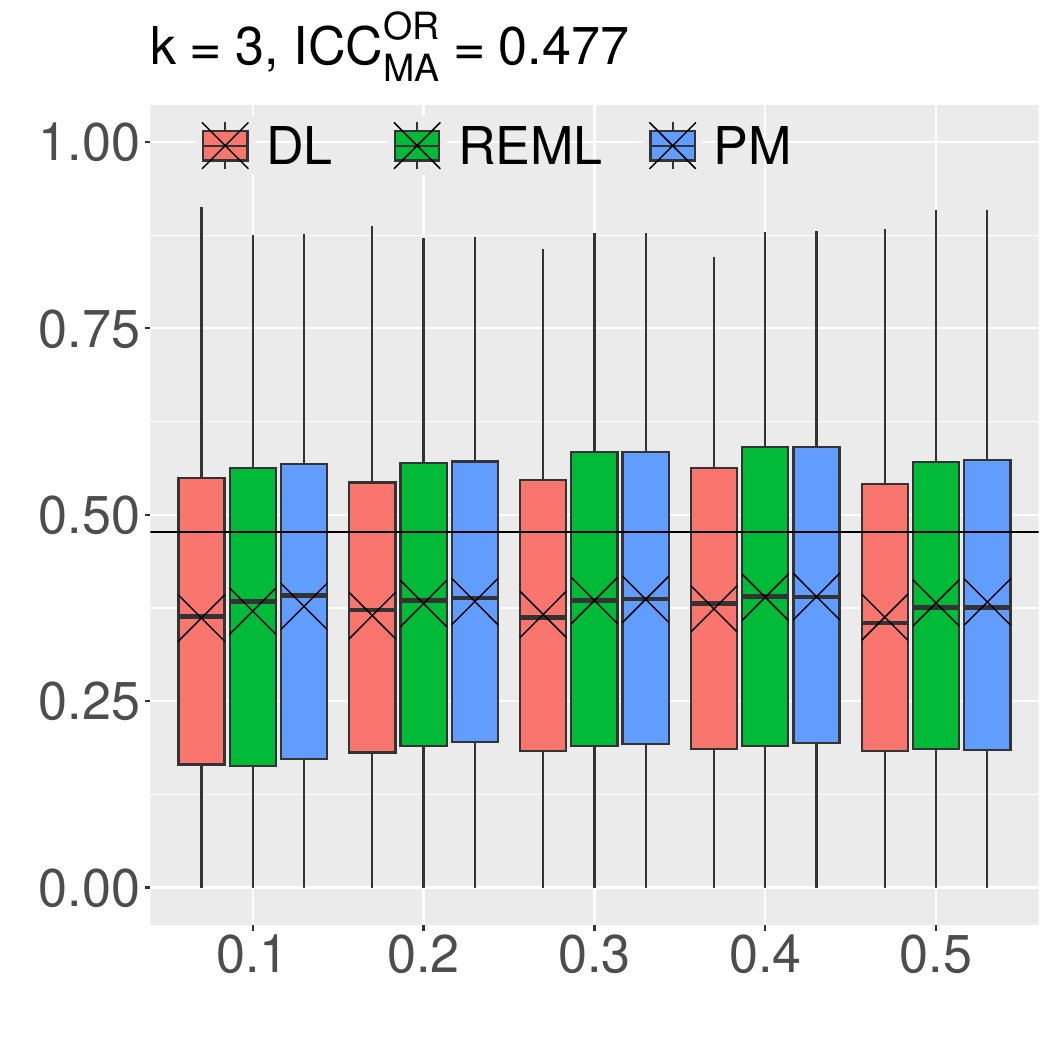,width=3.0in,angle=0}\\
			\psfig{figure=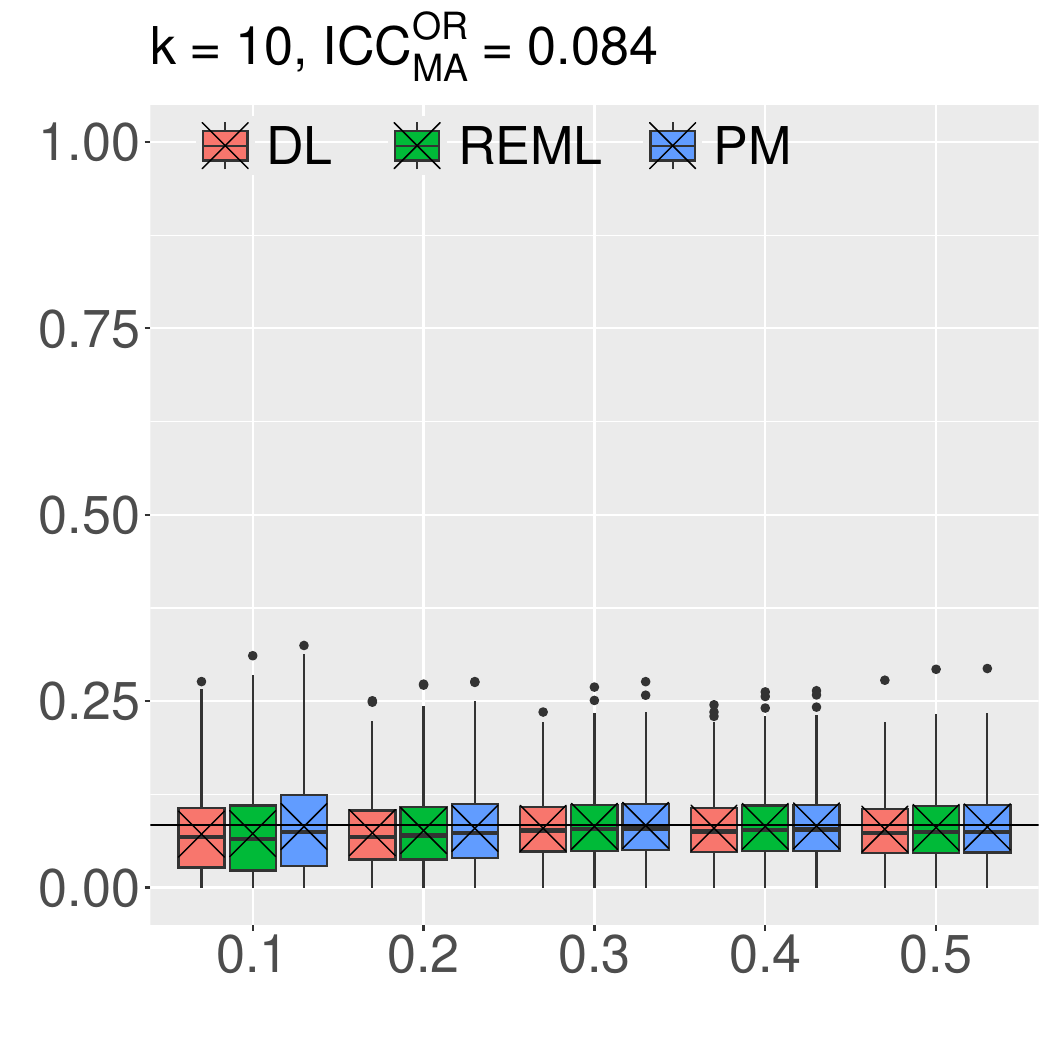,width=3.0in,angle=0}&
			\psfig{figure=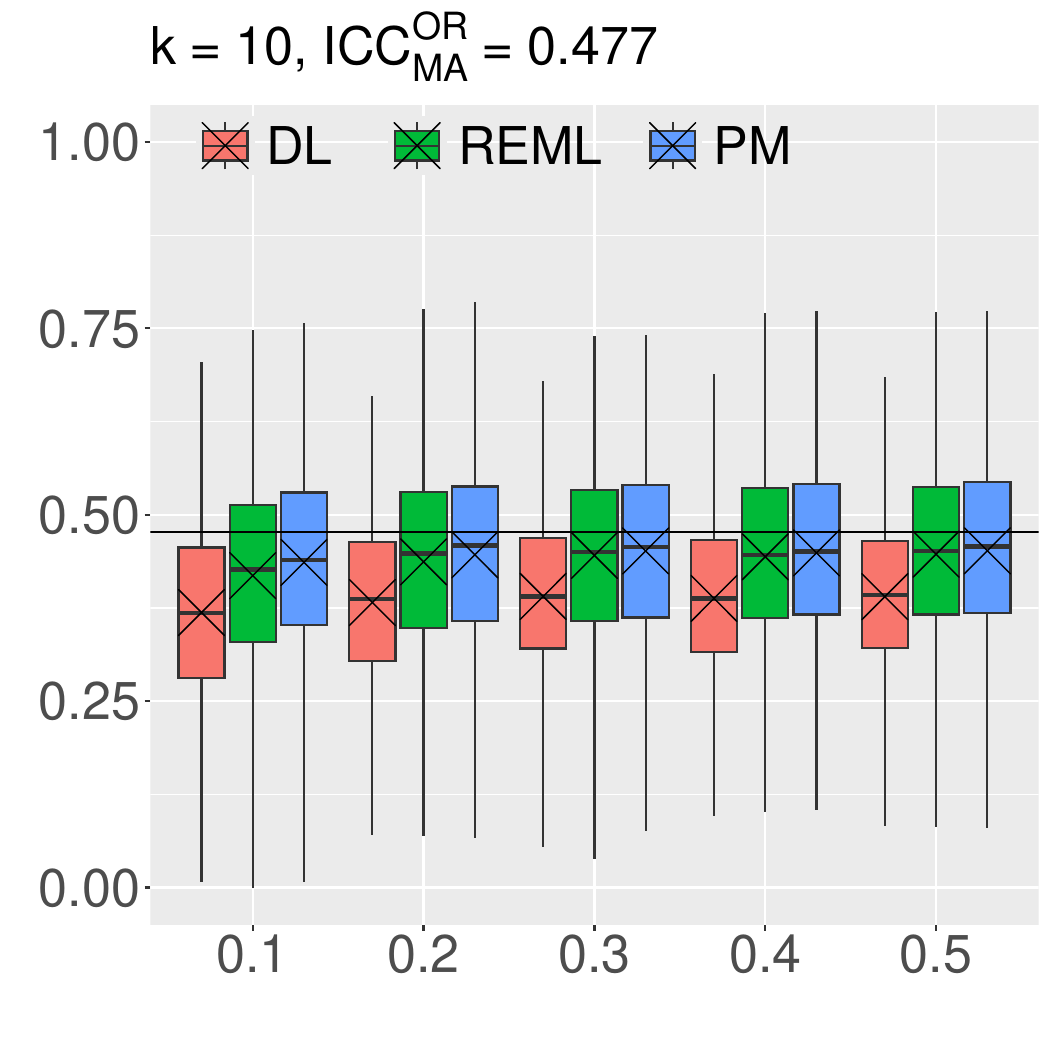,width=3.0in,angle=0}
		\end{tabular}
		\vspace{5mm}
		{\caption{Boxplots of the $I^2_{\rm A}$ statistics under the null-effect scenario, based on the DL (red boxes), REML (green boxes), and PM (blue boxes) estimators of $\tau^2$. Each cross indicates the mean of 1,000 replications, and the solid horizontal lines denote the true heterogeneity level ${\rm ICC}_{\rm MA}^{\rm OR}$.}\label{fig1}}
	\end{center}
\end{figure}
\begin{figure}[htp!]
	\begin{center}
		\begin{tabular}{cc}
			\psfig{figure=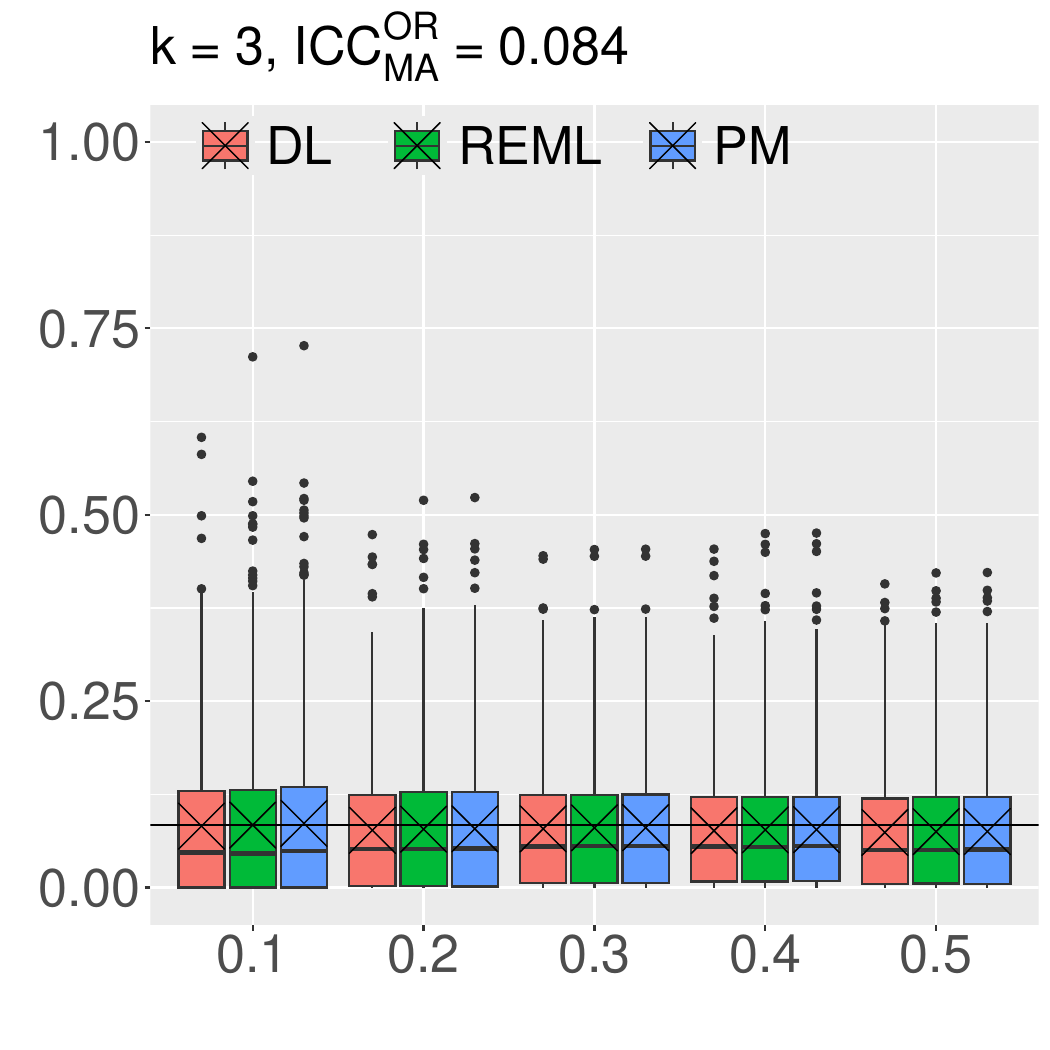,width=3.0in,angle=0}&
			\psfig{figure=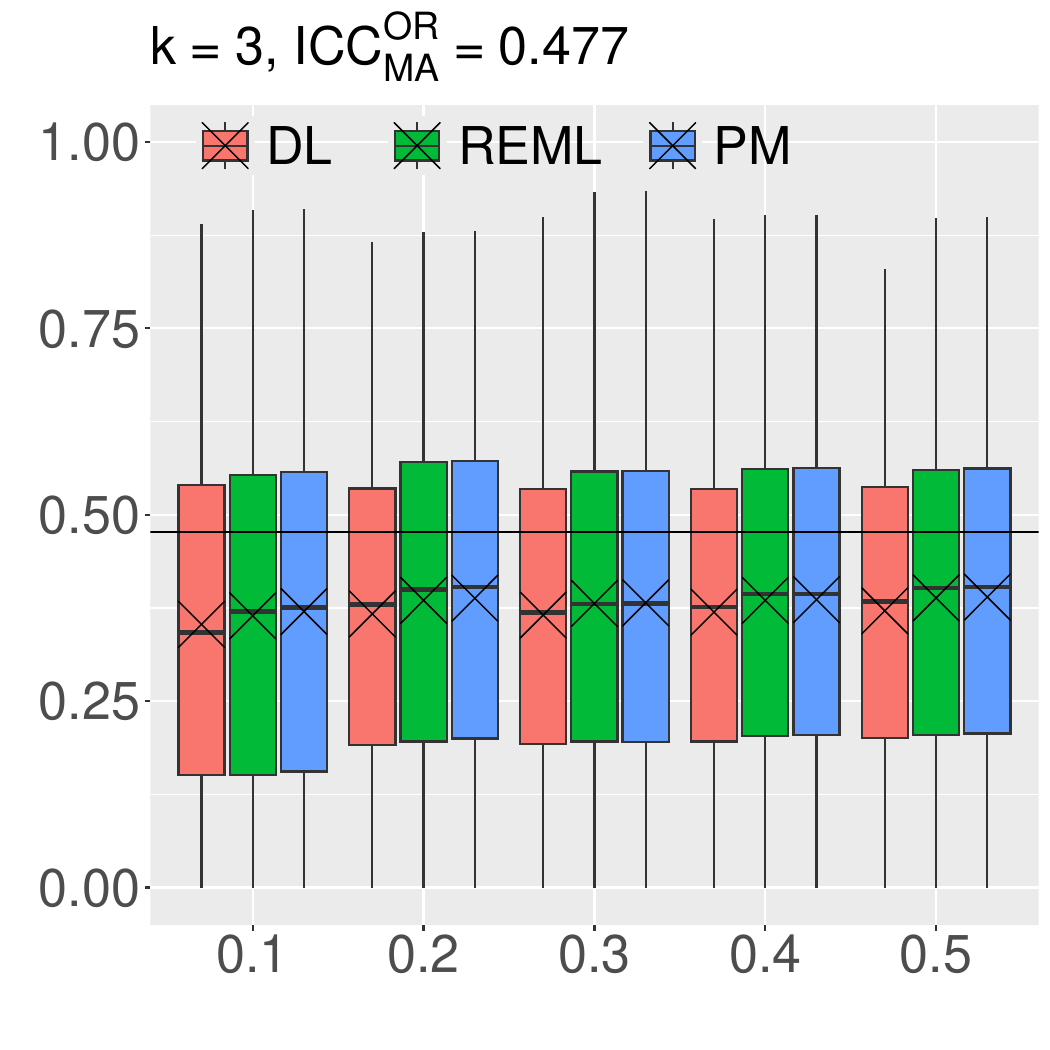,width=3.0in,angle=0}\\
			\psfig{figure=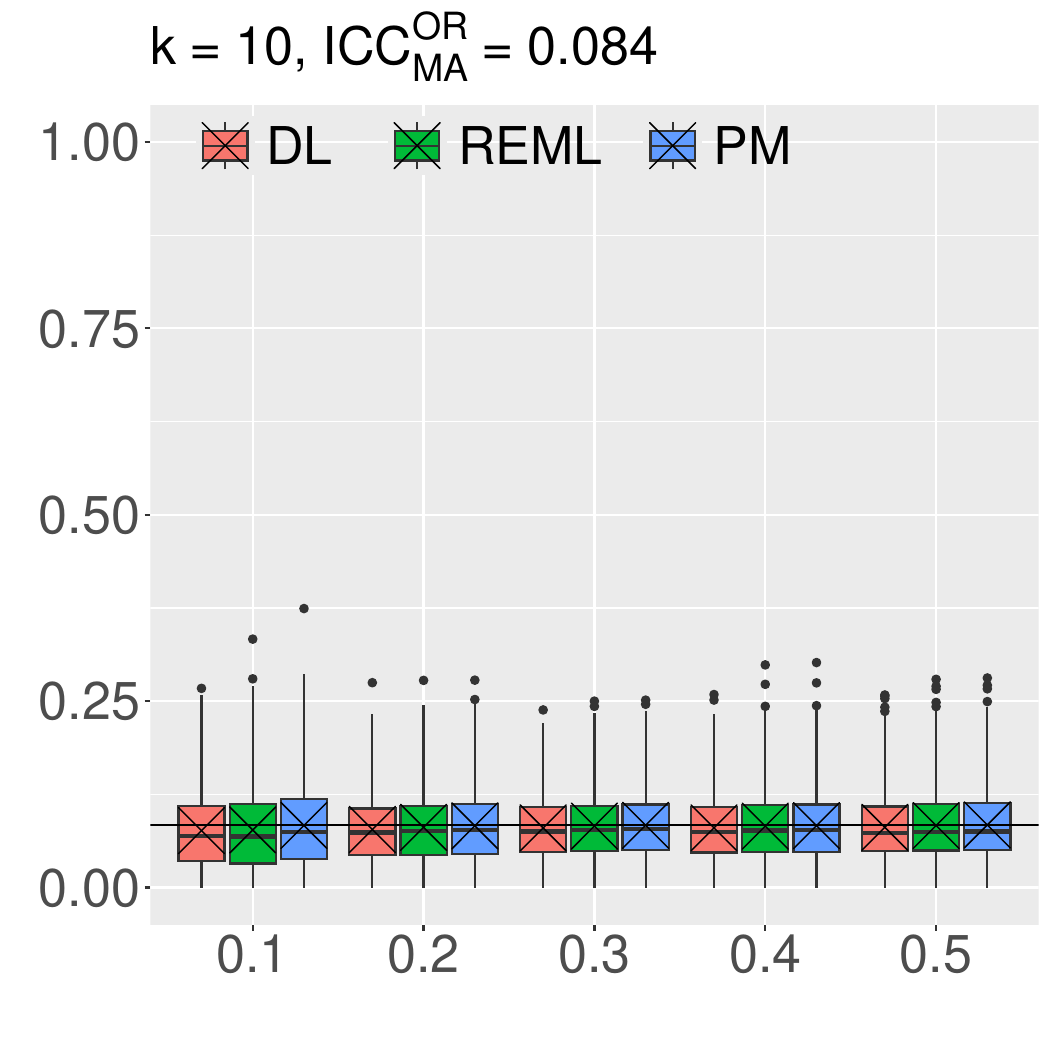,width=3.0in,angle=0}&
			\psfig{figure=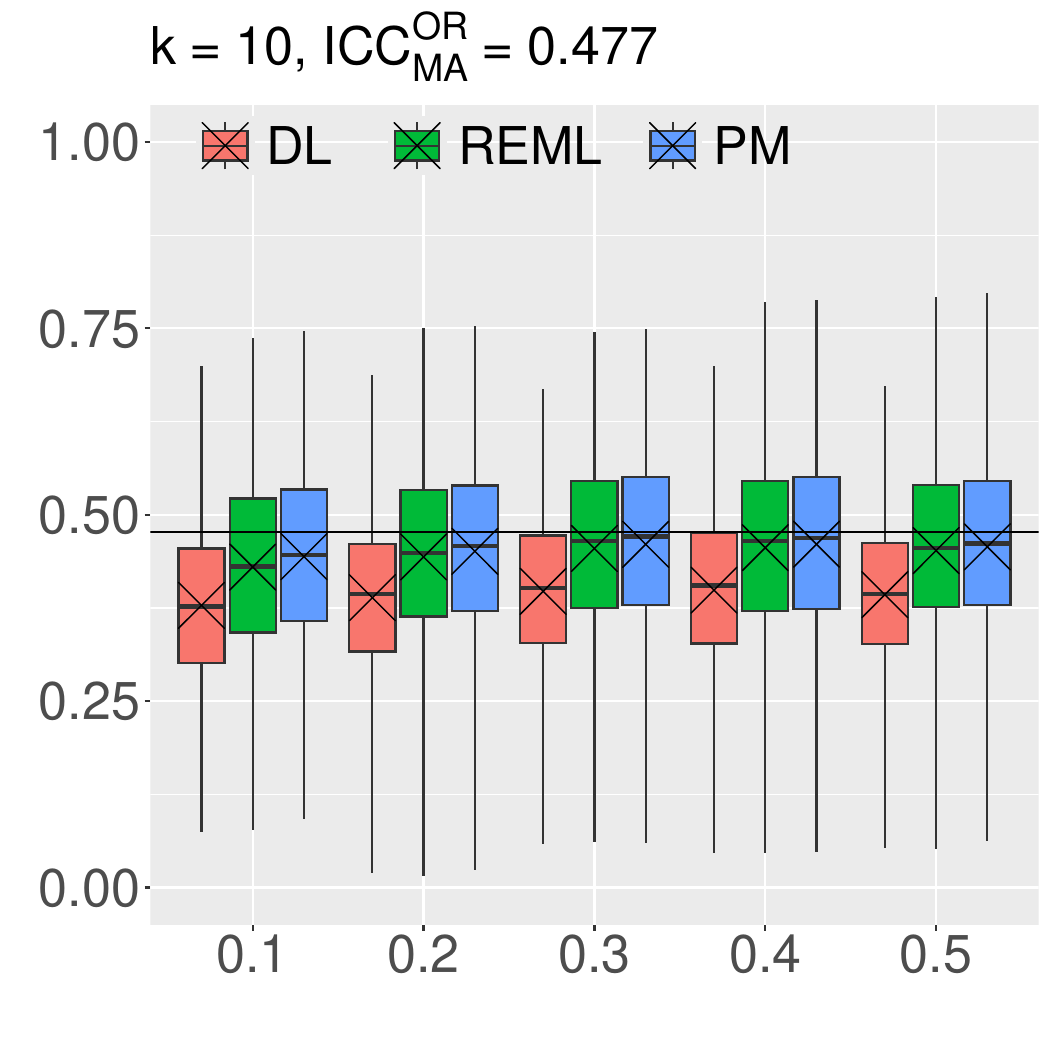,width=3.0in,angle=0}
		\end{tabular}
		\vspace{5mm}
		{\caption{Boxplots of the $I^2_{\rm A}$ statistics under the positive-effect scenario, based on the DL (red boxes), REML (green boxes), and PM (blue boxes) estimators of $\tau^2$. Each cross indicates the mean of 1,000 replications, and the solid horizontal lines denote the true heterogeneity level ${\rm ICC}_{\rm MA}^{\rm OR}$.}\label{fig2}}
	\end{center}
\end{figure}

From Figure \ref{fig1}, when the true between-study variance $\tau^2$ is small, the $I^2_{\rm A}$ statistic shows lower bias and variability. Increasing the number of studies from $k=3$ to $k=10$ further reduces both bias and variance, indicating improved estimation accuracy with larger study sample sizes. The overall event rate $p^C$ has only a minor impact, with slightly improved performance when $p^C$ is close to 0.5.
Among different estimators of $\tau^2$, the DL-based $I^2_{\rm A}$ generally exhibits smaller variance but larger bias than the REML- and PM-based versions, especially under large heterogeneity. Although its bias decreases as $k$ increases, it remains noticeable even when $k=10$, whereas the PM-based estimator performs best, yielding nearly unbiased estimates for larger $k$.
Similar patterns are observed in Figure \ref{fig2}, suggesting that the magnitude of the effect size has little influence on the performance of $I^2_{\rm A}$ in this setting.

\subsection{Consistency of $I^2_{\rm A}$ across latent and binary outcomes}

To further evaluate the consistency of $I^2_A$ between the latent outcomes and their corresponding binary outcomes, we conduct simulations based on the latent-variable model defined in equation (\ref{latent_model}).

For each study $i=1,\ldots,k$, latent outcomes for the control and treatment groups are generated with overall means $\tilde{\mu}^C=\tilde{\mu}^T=0$. Two levels of between-study heterogeneity are considered, characterized by $\tilde{\tau}^2=0.09$ and $0.9$, corresponding to ${\rm ICC}_{\rm MA}^{\rm LV}=0.083$ and $0.474$, respectively. The study-specific deviations are independently drawn as $\tilde{\delta}_i^C,\tilde{\delta}_i^T \sim N(0,\tilde{\tau}^2/2)$, yielding a study-specific treatment effect $\tilde{\mu}_i=\tilde{\mu}^T+\tilde{\delta}_i^T-(\tilde{\mu}^C+\tilde{\delta}_i^C)$ with variance $\tilde{\tau}^2$. Study-specific scaling is introduced via $\sigma_i \sim U(0.5,1.5)$.

We consider meta-analysis with $k=3$ or $k=10$ studies to represent small and large scenarios. For each study, the treatment and control arms have equal sample sizes $n_i^T=n_i^C=i\times n$, where $i=1,\ldots,k$ and $n\in\{10,\ldots,90\}$. Individual-level residual errors $\xi_{ij}^C$ and $\xi_{ij}^T$ are independently generated under two distributions: logistic (mean 0, variance 1) and standard normal. Binary outcomes are obtained by dichotomizing latent responses at a fixed threshold $C_i=0$ according to model (\ref{trunc}), and event counts are recorded for each arm.

For each simulated dataset, we compute the heterogeneity statistic $I^2_A$ for both continuous and binary outcomes. For continuous outcomes, study-level means and standard deviations are employed to obtain SMDs and their within-study variances via Hedges' $g$ (\citealp{hedges1981distribution}), while for binary outcomes, lnORs and their variances are calculated from the $2\times2$ tables. Random-effects meta-analysis is then performed separately for SMDs and lnORs using the PM estimator to obtain $\tilde{\tau}^2$ and $\tau^2$, respectively. The corresponding $I^2_A$ statistics are then computed based on (\ref{lv}) and (\ref{iccor}).

For each combination of $(k,\tilde\tau^2,n)$, we conduct $M=1000$ replications. Results are summarized in Figures \ref{fig3} and \ref{fig4} using Chauvenet-type boxplots. Each figure reports $I^2_A$ for both latent and binary outcomes to assess their agreement, with reference ${\rm ICC}_{\rm MA}$ values based on continuous outcomes.
\begin{figure}[htp!]
	\begin{center}
		\begin{tabular}{cc}
			\psfig{figure=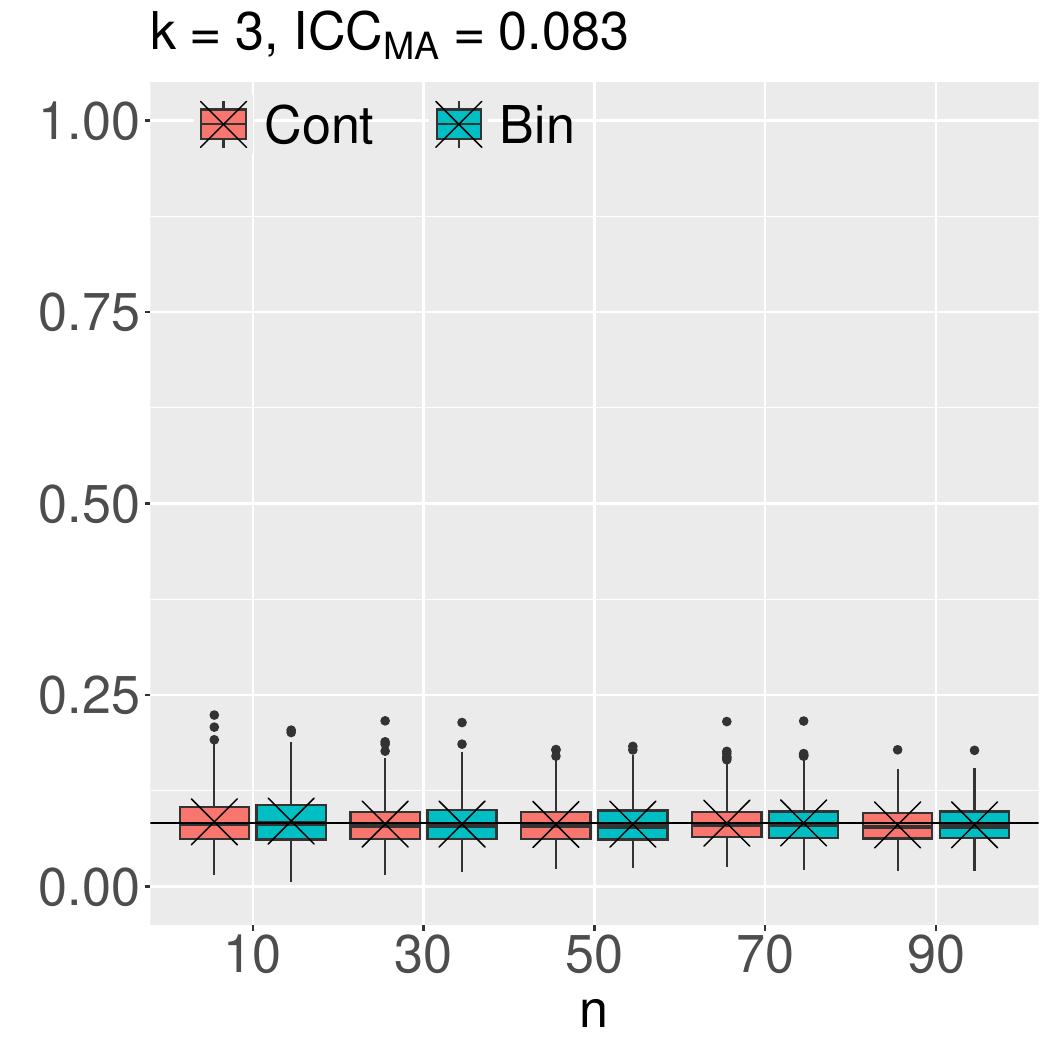,width=3.0in,angle=0}&
			\psfig{figure=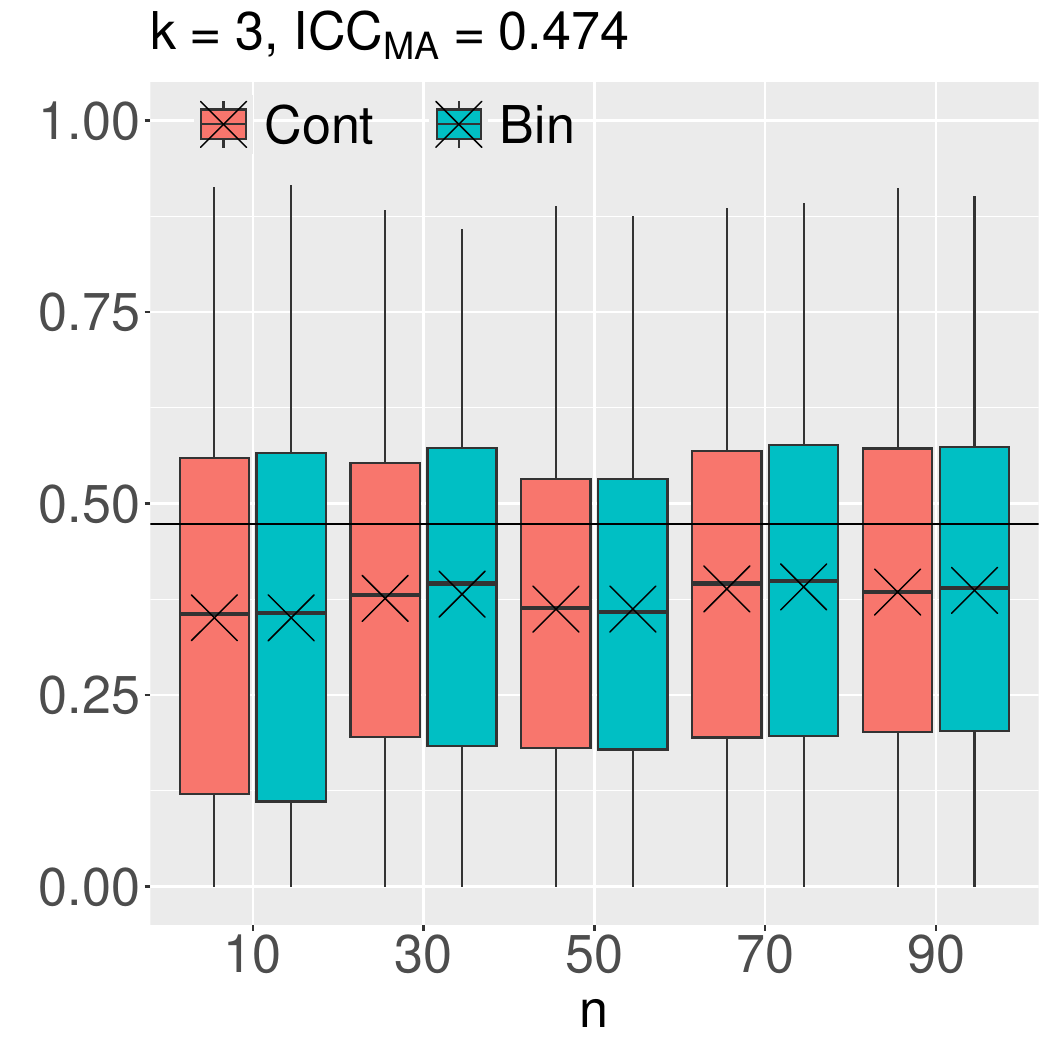,width=3.0in,angle=0}\\
			\psfig{figure=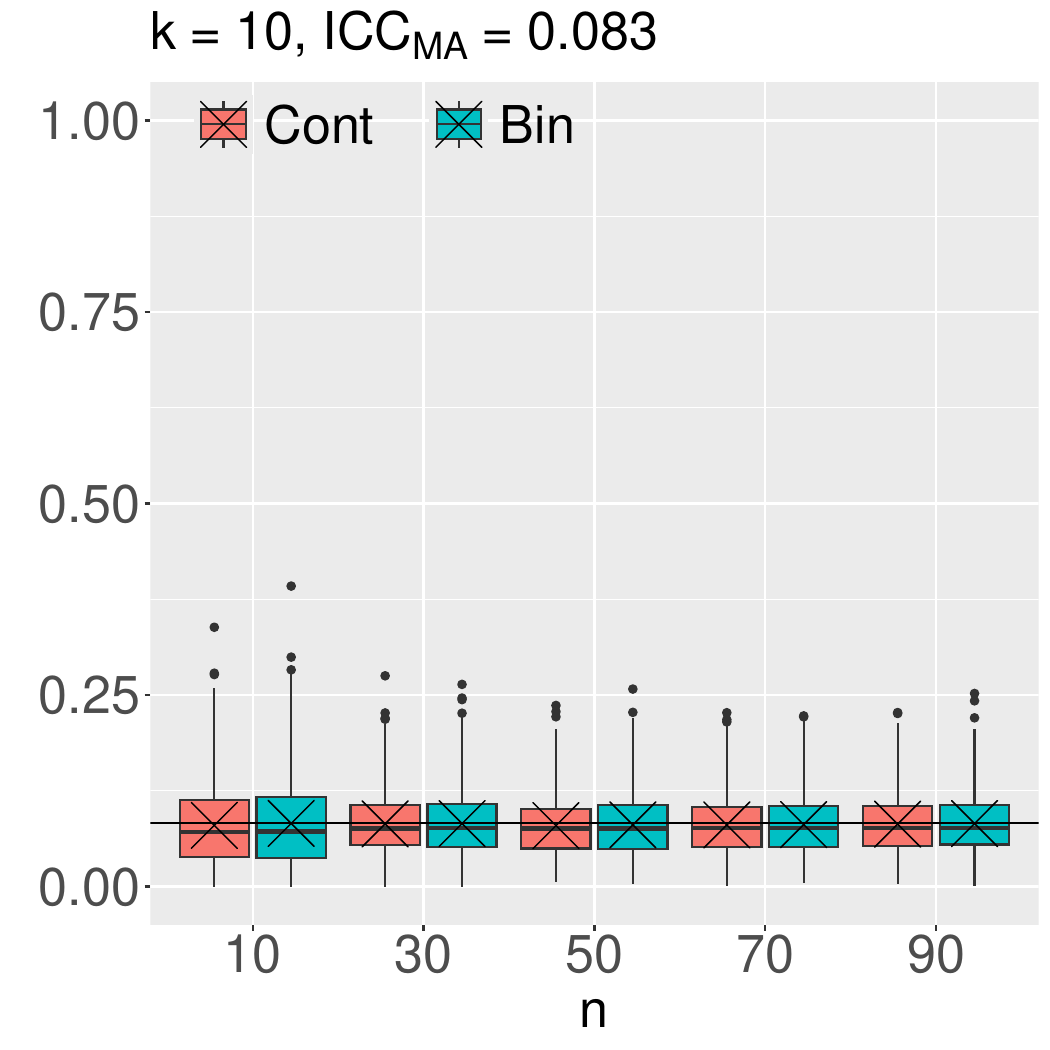,width=3.0in,angle=0}&
			\psfig{figure=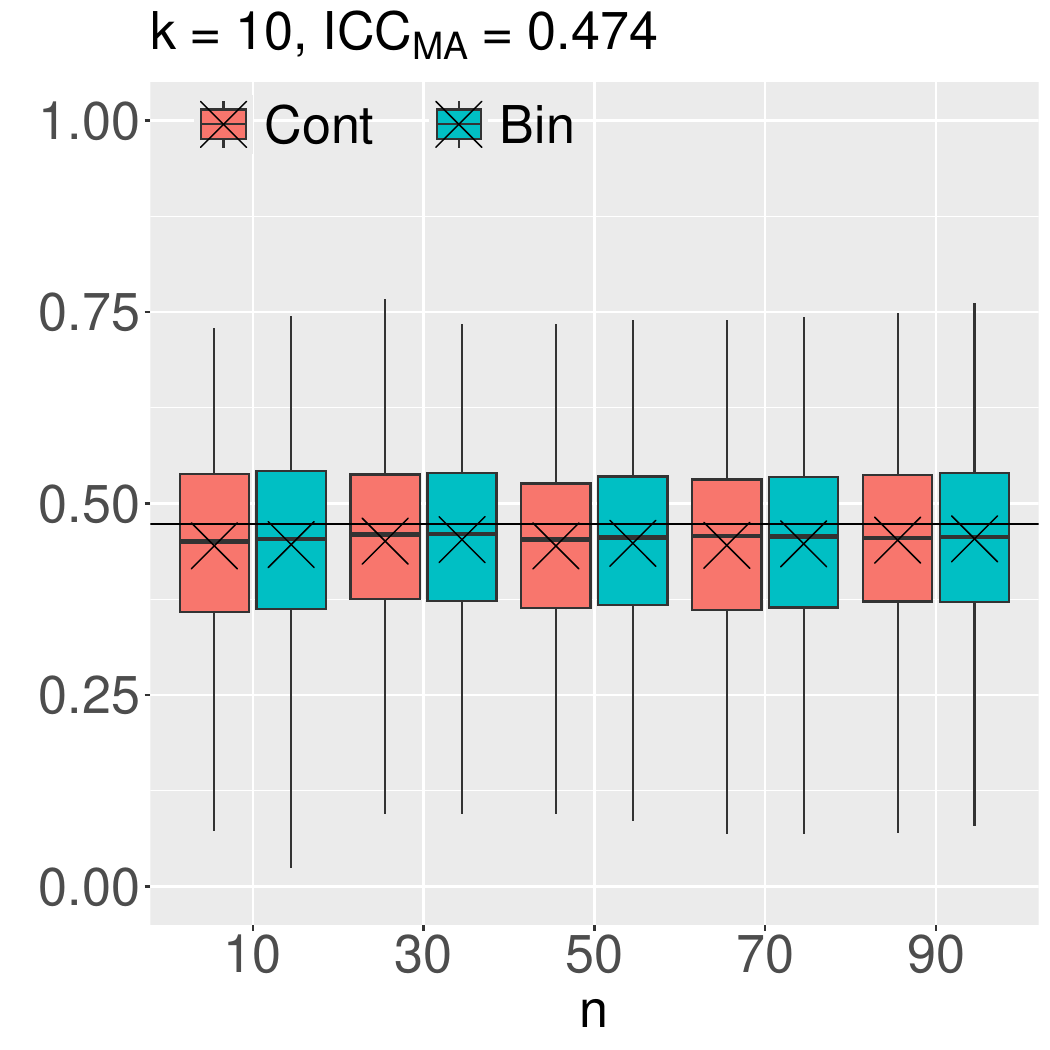,width=3.0in,angle=0}
		\end{tabular}
		\vspace{5mm}
		{\caption{Boxplots of the $I^2_{\rm A}$ statistics under the latent outcomes (red boxes) and the corresponding binary outcomes (blue boxes) with logistic residual errors. Each cross represents the mean across 1,000 replications, and the solid horizontal line denotes the true heterogeneity level ${\rm ICC}_{\rm MA}^{\rm LV}$.}\label{fig3}}
	\end{center}
\end{figure}
\begin{figure}[htp!]
	\begin{center}
		\begin{tabular}{cc}
			\psfig{figure=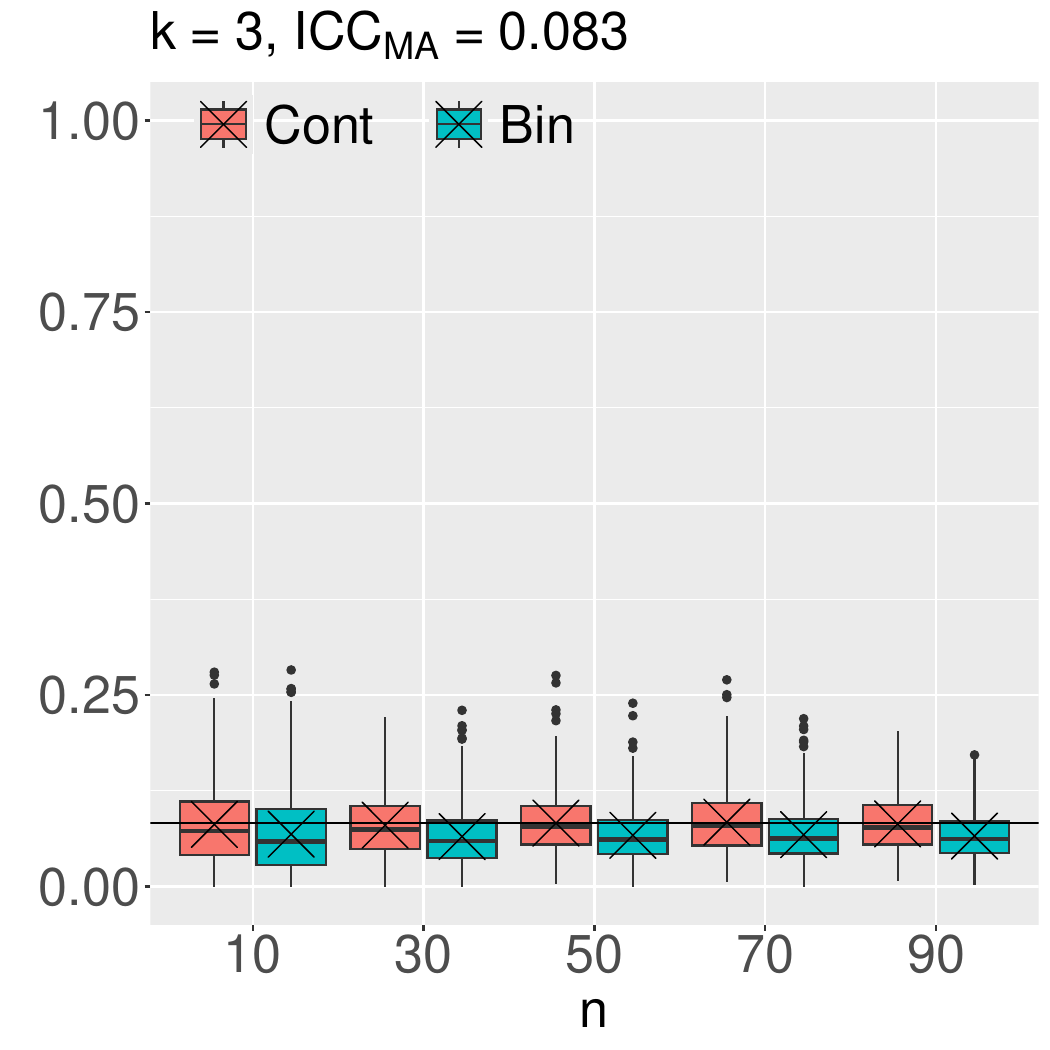,width=3.0in,angle=0}&
			\psfig{figure=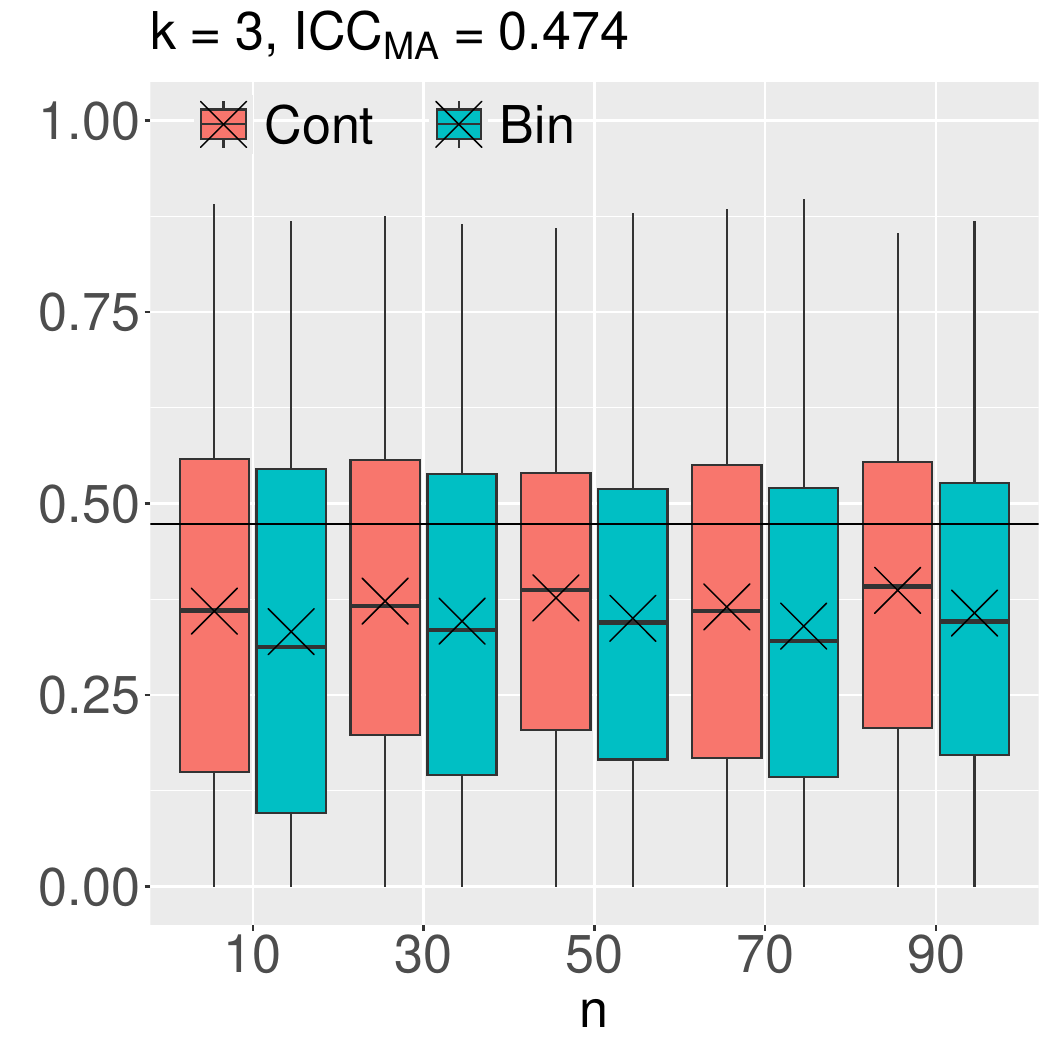,width=3.0in,angle=0}\\
			\psfig{figure=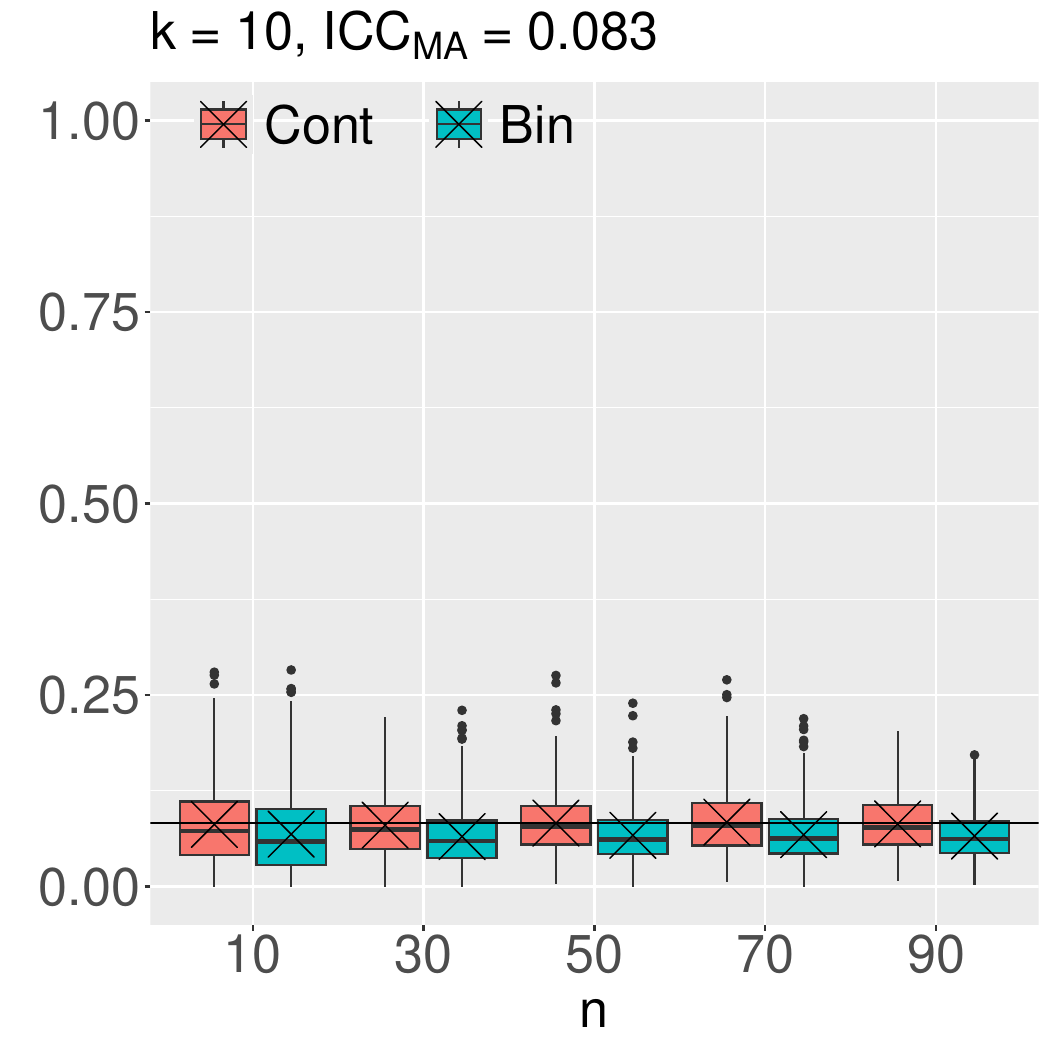,width=3.0in,angle=0}&
			\psfig{figure=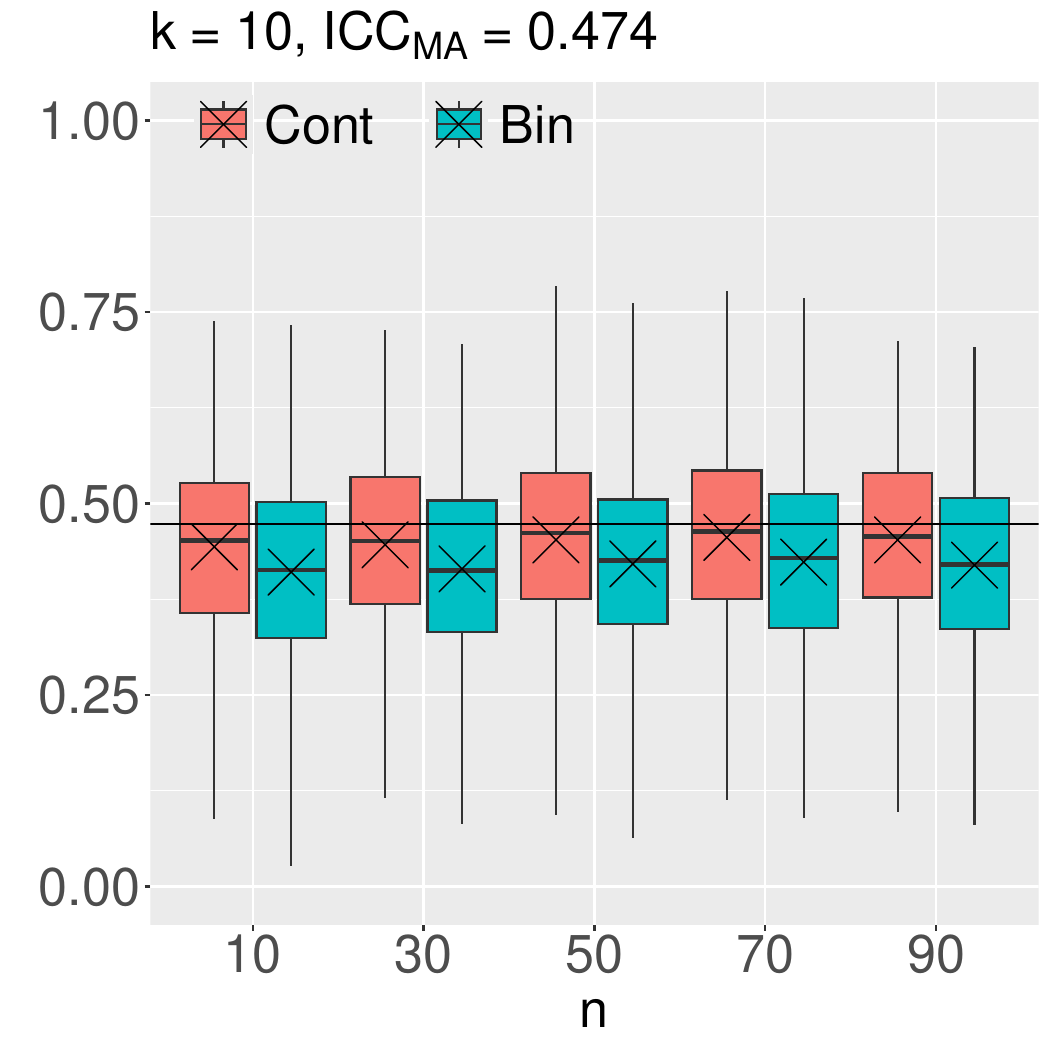,width=3.0in,angle=0}
		\end{tabular}
		\vspace{5mm}
		{\caption{Boxplots of the $I^2_{\rm A}$ statistics under the latent outcomes (red boxes) and the corresponding binary outcomes (blue boxes) with standard normal residual errors. Each cross represents the mean across 1,000 replications, and the solid horizontal line denotes the true heterogeneity level ${\rm ICC}_{\rm MA}^{\rm LV}$.}\label{fig4}}
	\end{center}
\end{figure}

Figure \ref{fig3} reports $I^2_A$ under logistic errors for continuous outcomes. The results show that $I^2_A$ for binary outcomes closely matches that for continuous outcomes across all settings. When the within-study sample size is small, the variance of $I^2_A$ for binary outcomes is slightly larger, but this difference diminishes as the number of studies increases. Overall, increasing $k$ yields nearly unbiased estimates for both outcome types, indicating consistent heterogeneity interpretation across binary and continuous outcomes.
Figure \ref{fig4} presents results under normal errors. Continuous outcomes show similar patterns as under logistic errors, while binary outcomes exhibit slightly larger bias but comparable variance. This bias decreases with $k$, though a small discrepancy remains. These findings suggest that even under normal errors, ${\rm ICC}_{\rm MA}$ for binary outcomes remains close to that of continuous outcomes, supporting the robustness and interpretability of $I^2_A$ across outcome types.

\subsection{Revisiting the motivating example}
In this section, we apply the proposed heterogeneity measure to reanalyze the motivating example in Section 2, which comprises five randomized controlled trials assessing the effectiveness of depression treatment. As each study reports both continuous and binary outcomes, this dataset allows a direct comparison of the heterogeneity estimates derived from different outcome types within the same meta-analytic framework.

For continuous outcomes, SMD is employed as the effect size. For each study, the observed SMD and its within-study variance are computed, followed by a random-effects meta-analysis using the PM estimator for $\tilde\tau^2$. The estimated between-study variance is $\widehat{\tilde\tau}^2(\rm SMD)=0.042$, leading to a heterogeneity statistic
\beqrs
I^2_A({\rm SMD})=\frac{0.042}{0.042+1}=0.040.
\eeqrs
For the binary outcomes, lnOR serves as the effect size. Using the observed lnORs and their within-study variances, the PM estimator yielded $\hat\tau^2({\rm lnOR})=0.088$, corresponding to
\beqrs
I^2_A({\rm lnOR})=\frac{0.088}{0.088+\pi^2/3}=0.026.
\eeqrs

To further examine the consistency between continuous and binary analyses, each study's SMD and its within-study variance are transformed to lnOR scale by the classical SMD-to-lnOR conversion formula: $\mu_i=\pi\tilde{\mu}_i/\sqrt{3}$ and $\sigma_{y_i}^2=\pi^2\sigma_{\tilde{y}_i}^2/3$,
where $\tilde\mu_i$ and $\sigma_{\tilde y_i}^2$ denote the study-specific effect size and its within-study variance on continuous outcome scale, while $\mu_i$ and $\sigma_{y_i}^2$ represent the corresponding transformed lnOR and its within-study variance for the $i$th study (\citealp{murad2019continuous}).
A subsequent random-effects meta-analysis based on the converted lnORs produced an estimated $\hat\tau^2({\rm clnOR})=0.138$, yielding
\beqrs
I^2_A({\rm clnOR})=\frac{0.138}{0.138+\pi^2/3}=0.040.
\eeqrs

Figure \ref{fig5} summarizes the results of the three random-effects meta-analyses. The left panel displays the study-specific effect sizes with their 95\% confidence intervals (CIs) for the three outcome types, with the last row showing the pooled overall estimates and corresponding CIs. The right panel presents the corresponding forest plot.
\begin{figure}[htbp]
	\centering
	\includegraphics[width=1\textwidth]{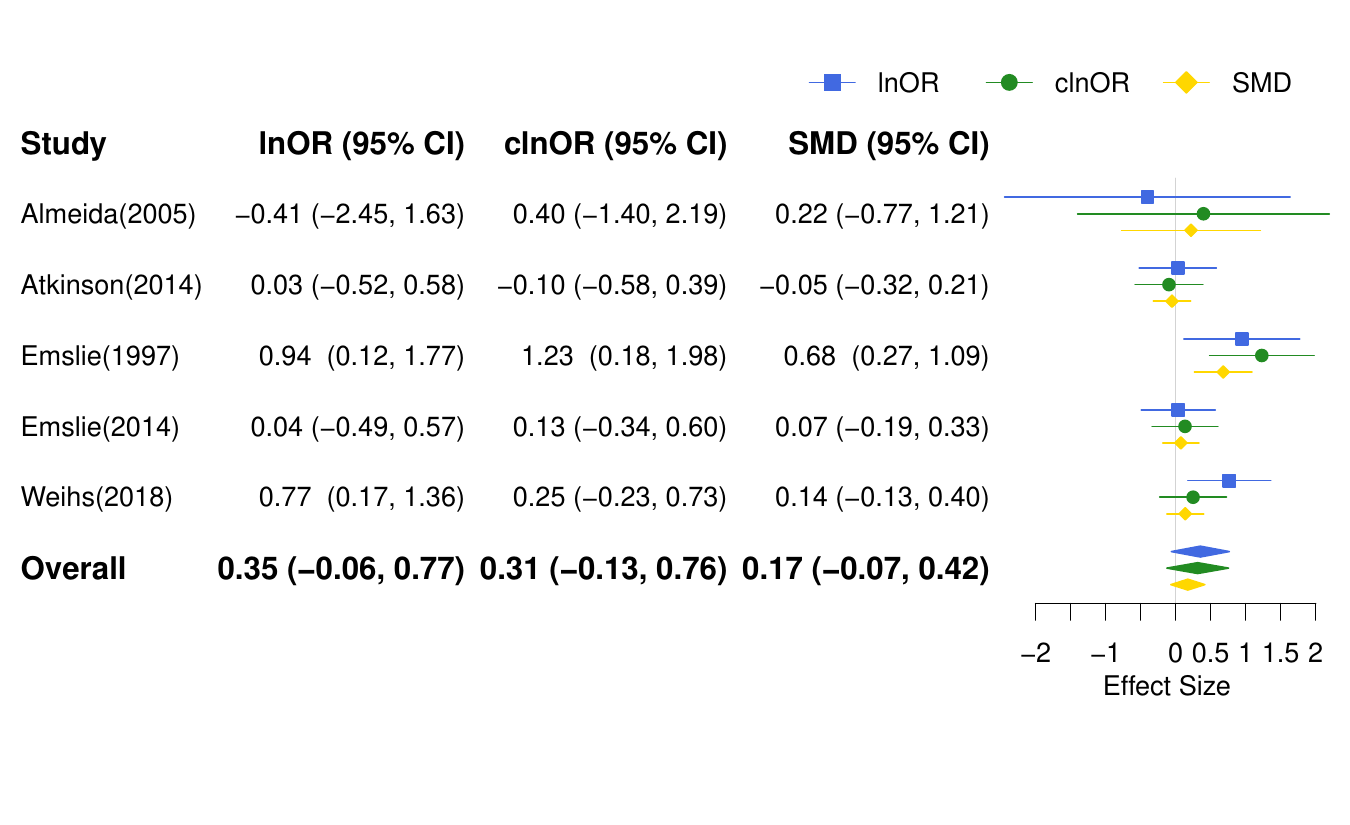}
	\caption{Results and forest plot of three random-effects meta-analyses based on different outcome types. lnORs are represented by blue squares, converted lnORs (clnOR) by green circles, and SMDs by yellow diamonds.}
	\label{fig5}
\end{figure}

From the meta-analytic results and forest plot, we observe that although SMD and the converted lnOR (clnOR) are expressed on different effect size scales, they are derived from the same underlying data. As a result, their between-study variance estimates, 0.0421 for SMD and 0.138 for clnOR, differ in magnitude; however, the corresponding heterogeneity measures, both equal to 0.040, remain identical. This finding indicates that the transformation of the effect size type does not influence the absolute heterogeneity captured by the proposed $I^2_A$.

In contrast, both lnOR and the converted lnOR are expressed on the same scale. The converted lnOR is obtained by transferring SMD onto lnOR metric, which facilitates a direct comparison with lnORs calculated from binary data. Although the two originate from different outcome measures, they both describe the same participants' level of depression, and thus are expected to exhibit a certain degree of consistency. The results confirm this expectation: despite minor differences in the study-level effect sizes, the estimation precision is comparable across studies, and the pooled overall estimates are similar in both magnitude and precision. Moreover, the between-study variances are both small, 0.088 for lnOR and 0.138 for the converted lnOR, resulting in similarly low heterogeneity measures of 0.026 and 0.040, respectively. These results further demonstrate the consistency and interpretability of the proposed $I^2_A$ across different effect size formulations.

\section{Conclusion and Discussion}

In meta-analysis, the heterogeneity reflects the extent to which the overall effect
represents the included studies. The presence of heterogeneity challenges the common effect assumption and influences model selection between the fixed-effect model and REM. Under REM, large heterogeneity implies that the pooled effect may not adequately represent individual studies. In such cases, further analyses, such as subgroup meta-analysis or meta-regression, are recommended to explore potential sources of heterogeneity (\citealp{higgins2019cochrane}). Recently, \cite{yang2025alternative} introduced an absolute measure of heterogeneity, ${\rm ICC}_{\rm MA}$, that ranges between 0 and 1. This measure is not influenced by study sample sizes and directly measures the heterogeneity between the study populations involved in the meta-analysis. More recently, \cite{yang2026five} proposed a six-level descriptive grading scheme for interpreting ${\rm ICC}_{\rm MA}$ and its practical estimator, $I_A^2$. Their approach, based on an ANOVA framework, focused on continuous outcomes.

In this paper, we developed the counterpart of ${\rm ICC}_{\rm MA}$ to meta-analysis with binary outcomes, for which lnOR is adopted as the effect size, and further proposed an absolute measure of the between-study heterogeneity,
\beqrs
{\rm ICC}_{\rm MA}^{\rm OR}=\frac{\tau^2}{\tau^2+\pi^2/3}.
\eeqrs
This measure is defined as the ratio of the between-study variance $\tau^2$ of the study-specific lnOR values to the sum of $\tau^2$ and a typical population variance $\pi^2/3$, which arises naturally from a latent variable model underlying the binary outcomes. Specifically, the combination of study-specific thresholds and scale parameters in the latent-variable formulation effectively standardizes the latent outcomes, ensuring that the population variance in the denominator is constant across meta-analysis and independent of the observed study data. This parallels the definition of ${\rm ICC}_{\rm MA}$ for continuous outcomes using SMD as the effect size, where the population variance is 1. In the current setting of lnOR, an additional scaling factor of $\pi/\sqrt{3}$ is introduced to account for the log-odds transformation. Consequently, the heterogeneity measure ${\rm ICC}_{\rm MA}^{\rm OR}$ for meta-analysis with binary outcomes is well defined, interpretable, and scale-invariant in terms of the latent outcomes.

Similar to ${\rm ICC}_{\rm MA}$ proposed for meta-analysis with continuous outcomes by \cite{yang2025alternative}, ${\rm ICC}_{\rm MA}^{\rm OR}$ possesses several desirable characteristics. Specifically, it is (a) monotonically increasing with respect to the between-study variance $\tau^2$, (b) invariant to the location of the effect sizes, (c) independent of the total number of studies included, and (d) unaffected by sample sizes of individual studies. These properties justify its designation as an absolute measure of heterogeneity. In practice, we can compute the corresponding heterogeneity statistic $I^2_A$ by substituting an estimator of $\tau^2$ into ${\rm ICC}_{\rm MA}^{\rm OR}$, yielding
\beqrs
I^2_A=\frac{\hat\tau^2}{\hat\tau^2+\pi^2/3}.
\eeqrs
Our simulation studies indicated that among widely applied estimators, the Paule-Mandel method provides the most reliable performance for estimating ${\rm ICC}_{\rm MA}^{\rm OR}$, yielding the PM-based $I^2_A$ for practical use. To facilitate interpretation, we adopt the same grading scheme proposed by \cite{yang2026five} for meta-analysis with continuous outcomes. In this classification, $I^2_A=0$ reflects no heterogeneity, while values of 0-0.20, 0.20-0.40, 0.40-0.60, 0.60-0.80, and above 0.80 correspond to low, moderate, substantial, severe, and extreme heterogeneity, respectively. This graded scheme allows researchers to assess the practical magnitude of heterogeneity when using $I^2_A$ for meta-analysis with binary outcomes.

Through extensive simulation experiments and an application to the motivating real-data example, we demonstrated that the proposed ${\rm ICC}_{\rm MA}^{\rm OR}$ and the corresponding $I^2_A$ statistic coincide with their counterparts for continuous outcomes based on SMD. Overall, our findings provide a practical and theoretically grounded approach to quantifying the heterogeneity in meta-analysis with binary outcomes, complementing existing measures for continuous outcomes and enhancing interpretability in evidence synthesis.

For meta-analysis with binary outcomes, the natural log of the relative risks (lnRR) is also frequently employed as an alternative effect size to lnOR. The two measures are related by the formula $\ln{\rm OR}=\ln{\rm RR}+\ln\left\{(1-p_0)+p_0\exp\left(\ln{\rm OR}\right)\right\}$,
where $p_0$ denotes the event rate in the control group (\citealp{zhang1998s}). When $p_0$ is small, lnRR and lnOR are approximately equal. In such cases, the proposed $I^2_A$ statistic can still be applied. However, for general cases, this approximation may not hold. Therefore, further research is needed to extend the $I^2_A$ statistic to meta-analysis with lnRR as the effect size, enabling a more general application of the absolute measure of heterogeneity with binary outcomes.

\bibliographystyle{apalike} 
\bibliography{ref} 

\newpage
\begin{appendices}
	\setcounter{figure}{0}
	\renewcommand{\thefigure}{S\arabic{figure}}
	\setcounter{table}{0}
	\renewcommand{\thetable}{S\arabic{table}}
	\setcounter{equation}{0}
	\renewcommand{\theequation}{S\arabic{equation}}
	\renewcommand\thesection{Appendix~\Alph{section}}
	
\section{Proof of Proposition \ref{prop1}}\label{appA}

\begin{proof}
	When $\xi_{ij}^T,\xi_{ij'}^C\sim {\rm Logistic}(0,\sqrt{3}/\pi)$, it follows that $\pi/\sqrt{3}\cdot\xi_{ij}^T,\pi/\sqrt{3}\cdot\xi_{ij'}^C\sim {\rm Logistic}(0,1)$. From the cumulative distribution function (CDF) of the standard logistic distribution, we have
	\beqrs
	p_i^T&=&{\rm Pr}\left\{y_{Tij}=1\right\}={\rm Pr}\left\{\tilde y_{Tij}>C_i\right\}\\
	&=&{\rm Pr}\left\{\frac{\pi}{\sqrt{3}}\xi_{ij}^T>\frac{\pi}{\sqrt{3}}\left\{\frac{C_i}{\sigma_i}-(\tilde\mu^T+\tilde\delta_i^T)\right\}\right\}\\
	&=&\frac{1}{1+\exp\left[\frac{\pi}{\sqrt{3}}\left\{-\left(\tilde\mu^T+\tilde\delta_i^T\right)+\frac{C_i}{\sigma_i}\right\}\right]}.
	\eeqrs
	Similarly,
	\beqrs
	p_i^C=\frac{1}{1+\exp\left[\frac{\pi}{\sqrt{3}}\left\{-\left(\tilde\mu^C+\tilde\delta_i^C\right)+\frac{C_i}{\sigma_i}\right\}\right]}.
	\eeqrs
	Therefore, the log odds ratio for the $i$th study can be expressed as
	\beqrs
	\mu_i=\ln\frac{p_i^T/\left(1-p_i^T\right)}{p_i^C/\left(1-p_i^C\right)}&=&\frac{\pi}{\sqrt{3}}\left\{\left(\tilde\mu^T+\tilde\delta_i^T\right)-\frac{C_i}{\sigma_i}\right\}-\frac{\pi}{\sqrt{3}}\left\{\left(\tilde\mu^C+\tilde\delta_i^C\right)-\frac{C_i}{\sigma_i}\right\}\nonumber\\
	&=&\frac{\pi}{\sqrt{3}}\tilde\mu_i.
	\eeqrs
	Consequently,
	\beqrs
	\tilde\tau^2={\rm Var}\left(\tilde\mu_i\right)=\frac{3}{\pi^2}{\rm Var}\left(\mu_i\right)=\frac{3}{\pi^2}\tau^2.
	\eeqrs
\end{proof}

\section{Proof of the approximated relationship between REMs for SMD and lnOR with normal errors}\label{appB}

\begin{proof}
	When $\xi_{ij}^T,\xi_{ij'}^C\sim N(0,1)$, from the CDF of the standard normal distribution,
	\beqrs
	p_i^T={\rm Pr}\left\{\tilde y_{Tij}>C_i\right\}={\rm Pr}\left\{\xi_{ij}^T>\frac{C_i}{\sigma_i}-(\tilde\mu^T+\tilde\delta_i^T)\right\}=\Phi\left(\tilde\mu^T+\tilde\delta_i^T-\frac{C_i}{\sigma_i}\right),
	\eeqrs
	where $\Phi(\cdot)$ is the CDF of standard normal distribution.
	Similarly,
	\beqrs
	p_i^C=\Phi\left(\tilde\mu^C+\tilde\delta_i^C-\frac{C_i}{\sigma_i}\right).
	\eeqrs
	Therefore, the log odds ratio for the $i$th study can be expressed as
	\beqrs
	\mu_i&=&\ln\frac{p_i^T/\left(1-p_i^T\right)}{p_i^C/\left(1-p_i^C\right)}=\ln\frac{p_i^T}{1-p_i^T}-\ln\frac{p_i^C}{1-p_i^C}\\
	&=&\ln\frac{\Phi\left(\tilde\mu^T+\tilde\delta_i^T-\frac{C_i}{\sigma_i}\right)}{1-\Phi\left(\tilde\mu^T+\tilde\delta_i^T-\frac{C_i}{\sigma_i}\right)}-\ln\frac{\Phi\left(\tilde\mu^C+\tilde\delta_i^C-\frac{C_i}{\sigma_i}\right)}{1-\Phi\left(\tilde\mu^C+\tilde\delta_i^C-\frac{C_i}{\sigma_i}\right)}.
	\eeqrs
	Using the commonly adopted approximation that replaces the standard normal CDF by the CDF of a standardized logistic distribution with the same mean and variance
	\beqr\label{approx}
	\Phi\left(x\right)\approx\frac{1}{1+\exp\left(-\sqrt{\pi^2/3}x\right)},
	\eeqr
	we obtain
	\beqr\label{approxe}
	\mu_i\approx\frac{\pi}{\sqrt{3}}\left(\tilde\mu^T+\tilde\delta_i^T-\frac{C_i}{\sigma_i}\right)-\frac{\pi}{\sqrt{3}}\left(\tilde\mu^C+\tilde\delta_i^C-\frac{C_i}{\sigma_i}\right)=\frac{\pi}{\sqrt{3}}\tilde\mu_i.
	\eeqr
	Hence,
	\beqrs
	\tilde\tau^2={\rm Var}\left(\tilde\mu_i\right)\approx\frac{3}{\pi^2}{\rm Var}\left(\mu_i\right)=\frac{3}{\pi^2}\tau^2.
	\eeqrs
\end{proof}

Below, we further examine the approximation performance of the logistic-normal relationship.
When the individual-level latent continuous outcome follows a normal distribution, the approximate relationship (11) between the between-study variances $\tau^2$ for binary outcomes and $\tilde{\tau}^2$ for latent SMD arises from the logistic approximation to the standard normal CDF in (\ref{approx}). Essentially, this approximation replaces the standard normal distribution of the latent continuous outcome with a logistic distribution having the same mean and variance. Let 
\beqrs
L(x)=\frac{1}{1+\exp\left(-\sqrt{\pi^2/3}x\right)}.
\eeqrs
Web Figure \ref{figa1} illustrates the accuracy of this approximation. The left panel compares the CDFs of $N(0,1)$ and ${\rm Logistic}(0,\sqrt{3}/\pi)$, while the right panel shows the pointwise difference $\Phi(x) - L(x)$. The maximum absolute difference occurs at $x =\pm 0.683$ and is only 0.02266, indicating that the logistic approximation is reasonably accurate.
\begin{figure}[htp!]
	\begin{center}
		\begin{tabular}{cc}
			\psfig{figure=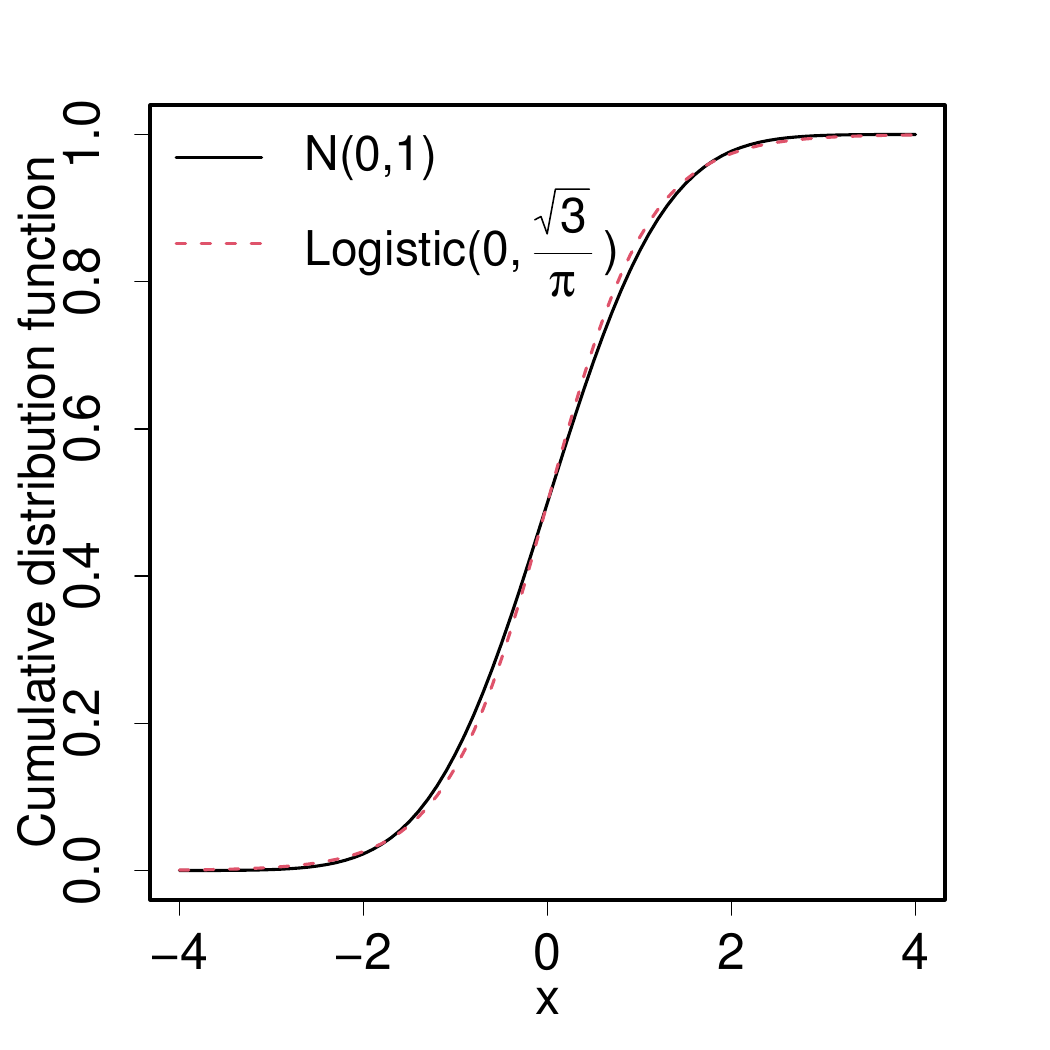,width=3.0in,angle=0}&
			\psfig{figure=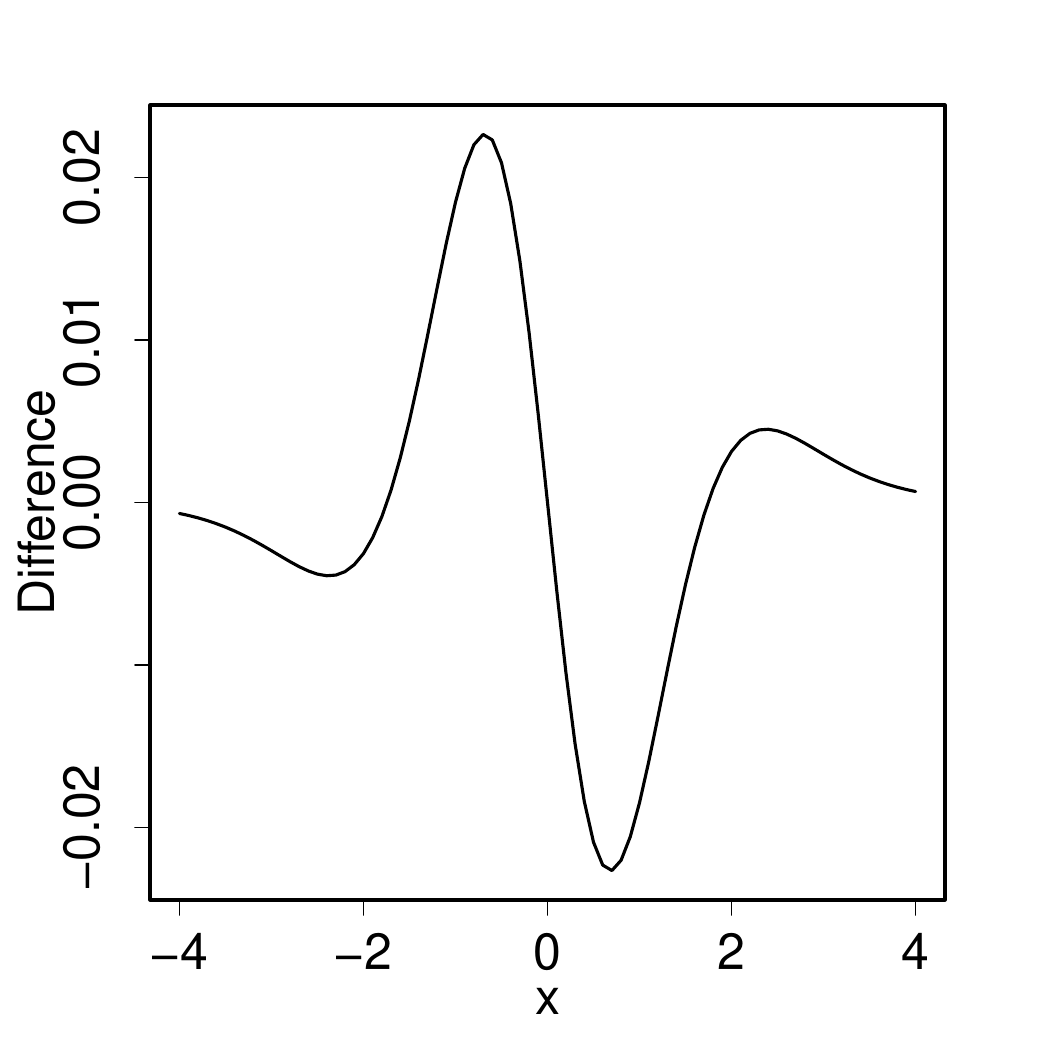,width=3.0in,angle=0}
		\end{tabular}{\caption{Approximation performance of the cumulative distribution function of $N(0,1)$ by ${\rm Logistic}(0,\sqrt{3}/\pi)$. Left panel: the CDFs of $N(0,1)$ and ${\rm Logistic}(0,\sqrt{3}/\pi)$; Right panel: the difference $\Phi(x) -L(x)$.}\label{figa1}}
	\end{center}
\end{figure}

We can further investigate the implication of this approximation at binary outcome level. For a binary outcome, the study-level effect size $\mu_i$ can be computed by the difference between the log odds of the treatment and control groups, $\ln\{p_i^T/(1-p_i^T)\}- \ln\{p_i^C/(1-p_i^C)\}$. According to the approximation in (\ref{approxe}), for any event probability $p$, the log odds $\ln\{p/(1-p)\}$ is approximated by $\pi/\sqrt{3}\cdot \Phi^{-1}(p)$. Web Figure \ref{figa2} illustrates the accuracy of this approximation across a range of event probabilities. The left panel compares the true log odds with the approximated log odds, and the right panel shows the pointwise difference. For event probabilities within the range 0.05 to 0.95 (corresponding to log odds between -2.944 and 2.944), which are not in the extreme tails, the maximum absolute difference occurs at $p = 0.158$ and $p = 0.842$, reaching 0.1455, indicating that the approximation is reasonably accurate for a wide range of practical purposes.
\begin{figure}[htp!]
	\begin{center}
		\begin{tabular}{cc}
			\psfig{figure=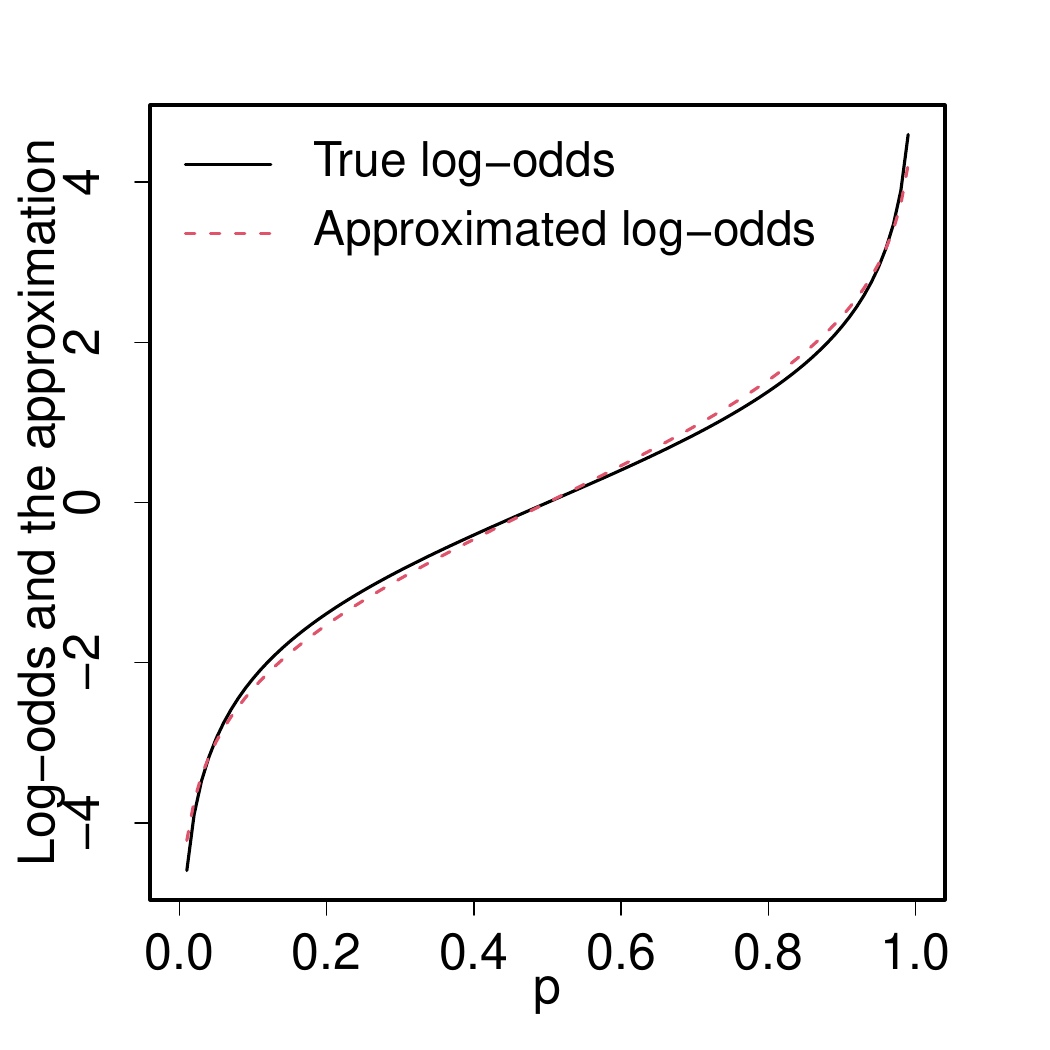,width=3.0in,angle=0}&
			\psfig{figure=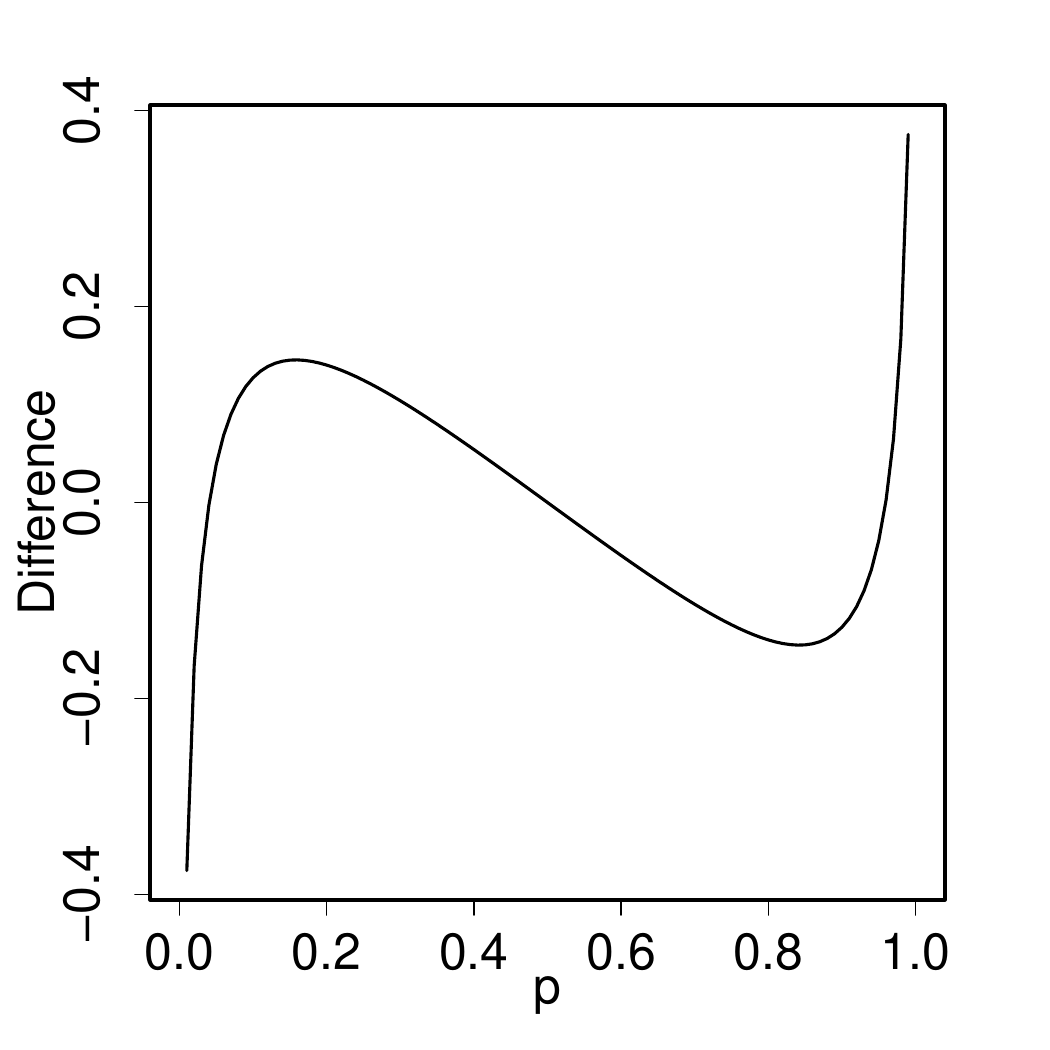,width=3.0in,angle=0}
		\end{tabular}{\caption{Approximation performance of the log odds at the binary outcome level. Left panel: true log odds $\ln\{p/(1-p)\}$ and the approximated log odds $\pi/\sqrt{3}\cdot \Phi^{-1}(p)$; Right panel: the difference $\ln\{p/(1-p)\} - \pi/\sqrt{3}\cdot \Phi^{-1}(p)$.}\label{figa2}}
	\end{center}
\end{figure}
\end{appendices}

\end{document}